\documentclass[compsoc, conference, a4paper, 10pt, times]{IEEEtran}

\usepackage{algorithm, algpseudocode}
\usepackage[dvipsnames]{xcolor}
\usepackage{soul}

\usepackage{epsfig, endnotes, xurl, graphicx}
\usepackage{amsmath, amsthm, amsfonts, mathrsfs}
\usepackage{mathtools, tikz}
\usepackage{diagbox, booktabs, colortbl}
\usepackage{hyperref,listings, framed}
\usepackage{multirow, caption, subcaption}
\usepackage{xspace, slashbox, enumitem, bm}
\DeclareMathAlphabet{\mathcal}{OMS}{cmsy}{m}{n}
\usepackage{bbding}
\usepackage{nicematrix, derivative}
\usepackage{wrapfig,lipsum,booktabs}
\usepackage{cleveref}
\usepackage{threeparttable}
\usepackage{tabularx}
\usepackage{pifont} 

\newcolumntype{Y}{>{\centering\arraybackslash}X} % for centering tables in tabularx

\newcommand\sysname{$\mathsf{chainBoost}$\xspace}
\newcommand\opBoost{$\mathsf{OPBoost}$\xspace}

\newcommand{\no}{\textcolor{Red}{\ding{54}}}
\newcommand{\yes}{\textcolor{OliveGreen}{\ding{52}}}

\newcommand{\ifcond}[1]{\textbf{if} {#1} \textbf{then}}
\newcommand{\forloop}[1]{\textbf{for} {#1} \textbf{do}}

\newcommand{\elsecond}{\textbf{else}}

\newcommand\sk{\mathsf{sk}\xspace}
\newcommand\pk{\mathsf{pk}\xspace}
\newcommand\GG{\mathbb{G}\xspace}
\newcommand\ZZ{\mathbb{Z}\xspace}

\newcommand\NN{\mathbb{N}\xspace}
\newcommand\cs{S_c\xspace}
\newcommand\sthresh{\theta_s\xspace}
\newcommand\lthresh{\theta_l\xspace}
\newcommand\com{C\xspace}
\newcommand{\all}{N}
\newcommand{\mal}{\mathscr{M}}

\newcommand\ppt{\mathsf{PPT}\xspace}

\newcommand\led{\mathcal{L}\xspace}
\newcommand\summ{\mathsf{sum}\xspace}
\newcommand\cid{\mathsf{cid}\xspace}
\newcommand\amount{\mathsf{amount}\xspace}
\newcommand\clients{\mathcal{C}\xspace}
\newcommand\servers{\mathcal{S}\xspace}
\newcommand\miners{\mathcal{M}\xspace}
\newcommand{\param}{\ensuremath{\mathtt{pp}}\xspace}
\newcommand\syssetup{\mathsf{SystemSetup}\xspace}
\newcommand\partysetup{\mathsf{PartySetup}\xspace}
\newcommand\stt{\mathsf{state}\xspace}
\newcommand\createTransaction{\mathsf{CreateTx}\xspace}
\newcommand\tx{\mathsf{tx}\xspace}
\newcommand\ask{\mathsf{ask}\xspace}
\newcommand\offer{\mathsf{offer}\xspace}
\newcommand\agreement{\mathsf{agreement}\xspace}
\newcommand\transfer{\mathsf{transfer}\xspace}
\newcommand\dispute{\mathsf{dispute}\xspace}
\newcommand\sync{\mathsf{sync}\xspace}
\newcommand\serviceProof{\mathsf{serviceProof}\xspace}
\newcommand\servicePayment{\mathsf{servicePayment}\xspace}
\newcommand\aux{\mathsf{aux}\xspace}
\newcommand\verifyTransaction{\mathsf{VerifyTx}\xspace}
\newcommand\updateState{\mathsf{UpdateState}\xspace}
\newcommand\mainc{\mathsf{mc}\xspace}
\newcommand\sidec{\mathsf{sc}\xspace}
\newcommand\leader{\mathsf{leader}\xspace}
\newcommand\elect{\mathsf{Elect}\xspace}
\newcommand\createSyncTransaction{\mathsf{CreateSyncTx}\xspace}
\newcommand\verifySyncTransaction{\mathsf{VerifySyncTx}\xspace}
\newcommand\verifyBlock{\mathsf{VerifyBlock}\xspace}
\newcommand\btype{\mathsf{btype}\xspace}
\newcommand\setup{\mathsf{Setup}\xspace}
\newcommand\meta{\mathsf{meta}\xspace}
\newcommand\summary{\mathsf{summary}\xspace}
\newcommand\txtype{\mathsf{txtype}\xspace}
\newcommand\head{\mathsf{head}\xspace}
\newcommand\block{\mathsf{B}\xspace}
\newcommand\prune{\mathsf{Prune}\xspace}

\newtheorem{theorem}{Theorem}
\newtheorem{lemma}{Lemma}
\newtheorem{definition}{Definition}

\begin{document}

\title{chainBoost: A Secure Performance Booster for Blockchain-based Resource Markets}

%\author{Anonymous submission to EuroS\&P'24}
\author{\IEEEauthorblockN{Zahra Motaqy}
\IEEEauthorblockA{\textit{University of Connecticut} \\
raha@uconn.edu}
\and
\IEEEauthorblockN{Mohamed E. Najd}
\IEEEauthorblockA{\textit{University of Connecticut} \\
menajd@uconn.edu}
\and
\IEEEauthorblockN{Ghada Almashaqbeh}
\IEEEauthorblockA{\textit{University of Connecticut} \\
ghada@uconn.edu}
%% IEEE format can accommodate up to six authors this way
}

% make the title area
\maketitle

\begin{abstract}
Cryptocurrencies and blockchain technology provide an innovative model for reshaping digital services. Driven by the movement toward Web 3.0, recent systems started to provide distributed services, such as computation outsourcing or file storage, on top of the currency exchange medium. By allowing anyone to join and collect cryptocurrency payments for serving others, these systems create \emph{decentralized markets} for trading digital resources. Yet, there is still a big gap between the promise of these markets and their practical viability. Existing initiatives are still early-stage and have already encountered security and efficiency obstacles. At the same time, existing work around promising ideas, specifically sidechains, fall short in exploiting their full potential in addressing these problems.

To bridge this gap, we propose \sysname, a secure performance booster for decentralized resource markets. It expedites service related operations, reduces the blockchain size, and supports flexible service-payment exchange modalities at low overhead. At its core, \sysname employs a sidechain, that has a (security and semantic) mutual-dependence with the mainchain, to which the system offloads heavy/frequent operations. To enable it, we develop a \emph{novel sidechain architecture} composed of temporary and permanent blocks, a \emph{block suppression} mechanism to prune the sidechain, a \emph{syncing protocol} to permit arbitrary data exchange between the two chains, and an \emph{autorecovery protocol} to support robustness and resilience. We analyze the security of \sysname, and implement a proof-of-concept prototype for a distributed file storage market as a use case. For a market handling around 2000 transactions per round, our experiments show up to 11x improvement in throughput and 94\% reduction in confirmation time. They also show that \sysname can reduce the main blockchain size by $\sim$90\%, and that it outperforms comparable optimistic rollup solutions by reducing transaction finality by 99.7\%. 
\end{abstract}

\section{Introduction}
\label{sec:intro}
Cryptocurrencies and blockchains provide an innovative model that led to new research frontiers in distributed computing and cryptography~\cite{Abadi20,Goyal20,Bentov14}. Driven by the movement toward the decentralized Internet---Web 3.0, recent cryptocurrency systems started to provide distributed services, such as computation outsourcing, content distribution, or file storage, on top of the currency exchange medium~\cite{filecoin,livepeer}. By allowing anyone to join and collect payments for serving others, these systems create \emph{decentralized markets} for trading digital resources~\cite{login}.\footnote{We focus on open-access decentralized resource markets that employ permissionless blockchains. For brevity, we refer to these as blockchain-based resource markets going forward.}

Blockchain-based resource markets improve early designs of peer-to-peer (P2P) systems by utilizing the properties of blockchains, namely, decentralization, public verifiability, and immutability. They emerged to address the trust, cost, and transparency concerns related to centrally-managed services. For example, in content distribution networks (CDNs), studies~\cite{Anjum17,Karamshuk15} showed that up to 88\% of the network traffic can be offloaded to P2P-based CDNs during peak demands, which improves performance and reduces cost. Others~\cite{content-based} advocated for building a decentralized content-addressed web that is more robust than the current location-based one due to evidence of lost files (a Harvard-led study found that 48\% of all hyperlinks cited in US Supreme Court opinions were broken~\cite{zittrain2014perma}). 

Additionally, blockchain-based resource markets create equitable ecosystems so end users can utilize their excess resources to collect revenue. Moreover, they provide new insights to improve blockchain design. Specifically, they promote the concept of useful mining, in which miners are selected to extend the blockchain based on the amount of service they provide~\cite{filecoin}, rather than based on wasteful computations as in Bitcoin. These potential benefits encouraged the development of many distributed resource markets in practice, such as Filecoin for file storage~\cite{filecoin}, Livepeer for video transcoding~\cite{livepeer}, and Golem for computation outsourcing~\cite{golem}, to name a few. Such systems are viewed as a basic component of Web 3.0 in which centrally-managed services that we currently use are transformed into fully-decentralized ones. 

\begin{table*}[t!]
\caption{Comparison of this work (\sysname), optimistic rollups (Optimism), and ZK rollups (ZkSync Era). Batch period for Optimism and ZkSync are between 30$\sim$60 sec, and 30$\sim$90 sec, respectively~\cite{jumpcryptoBridgingFinality,zksync-era-bf}.}
\vspace{-3pt}
\begin{threeparttable}

\begin{tabular}{|l|l|l|l|l|l|l|l|l|}
\hline
Solution\tnote{*} & Finality & \begin{tabular}[c]{@{}l@{}}  Throughput \end{tabular} & Decentralized &\begin{tabular}[c]{@{}l@{}}  Public \\ Verifiability \end{tabular} &\begin{tabular}[c]{@{}l@{}}  No Verifier\\ Dilemma \end{tabular} & \begin{tabular}[c]{@{}l@{}}  No Trusted \\ Setup\end{tabular} & \begin{tabular}[c]{@{}l@{}} No Smart Contract \\ Dependency \end{tabular} \\
\hline
Optimism~\cite{armstrong2021ethereum,op-verifier,optimism-tx}   & 7 days     & max 11 tx/sec      & \no & \no & \no & \yes & \no  \\
\hline
ZkSync Era~\cite{zksync-era-finality,zksyncWelcomeDocs}                  & 24 hours    & $8-25$ tx/sec & \no           & \yes                 & \yes                 & \no                      & \no                               \\
\hline
\sysname  & $2-30$ min               & $194-776$ tx/sec                                                   & \yes          & \yes                 & \yes                 & \yes                    & \yes                             \\
\hline
\end{tabular}
\begin{tablenotes}
\item[*] For Optimism, we show the maximum observed throughput~\cite{optimism-tx}. Finality means the duration needed to consider state changes induced by rollups or summaries final on the mainchain. In addition to batch processing, in Optimism, this period covers the contestation period; for ZKSync Era, it includes proof generation and verification times; and for \sysname, it is the time needed for transactions to appear in a meta-block. Lastly, \sysname numbers are from the experiments in Section~\ref{sec:perf-eval} for 0.5 and 2 MB sidechain block size and average batch-equivalent duration throughput.
\end{tablenotes}
\end{threeparttable}
\label{tbl:comparison}
\vspace{-6pt}
\end{table*}

\textbf{Challenges.} Unfortunately, there is still a big gap between the promise of blockchain-based resource markets and their practical viability. Blockchain protocols are complex applied cryptographic protocols, and complexity usually breeds vulnerabilities. Dealing with unauthenticated and untrusted participants, and involving monetary incentives, result in several security issues and attacks~\cite{bag2016bitcoin,conti2018survey,mccorry2018smart}. Resource markets are even more involved due to requiring complex modules to allow these parties to trade resources with each other. Not to mention that these markets must meet particular performance requirements, depending on the service type, in order to be a successful replacement to legacy, centrally-managed digital services.  

Most existing resource market initiatives are still early-stage and have already stumbled upon several security and efficiency obstacles. The countermeasures needed to address these threats, such as resource expenditure proofs~\cite{Moran19,Fisch19}, dispute solving, or market matching, introduce extra overhead adding to the scalability issues of blockchains. In many occasions, this led to deployments that favor efficiency at the expense of security; system designers may choose to forgo deploying important security countermeasures to optimize performance, leading to exploitable resource markets, e.g., avoid employing service delivery proofs, which allows servers to collect payments for free~\cite{nucypherBug}. Resolving core design issues at a later stage is also problematic as it leads to hard forks.

\textbf{Limitations of existing solutions.} Although several lines of work emerged to address the performance issues of blockchains, none of them are suitable for resource markets and, if used in this context, impose several limitations. Changing the blockchain parameters, such as increasing the block size, is problematic; it does not solve the scalability issues as all the data is still logged on the mainchain, and thus, is often unfavorable by the community  (e.g., the case of Bitcoin SV hard fork~\cite{bsv}). Micropayment schemes and payment channels~\cite{Decker15, Pass15} improve throughput of currency transfers rather than service data and events. While other layer-two solutions are either not fully decentralized due to the reliance on trusted verifiers (e.g., Optimisic and ZK Rollups~\cite{Poon17,starkware}), or slow due to the extra overhead of involving expensive zero-knowledge (ZK) proofs or repetitive computation and contesting process. We summarize the key differences between existing solutions and \sysname in Table~\ref{tbl:comparison}.

Others developed service-type specific techniques (e.g., lower overhead proof-of-storage~\cite{Fisch18}), or outsourced the work to third-party systems (e.g., TrueBit~\cite{truebit}) that are still early-stage and have their own challenges. These directions produced independent solutions that target specific use cases or systems with given semantics. Thus, there is no clear path on how to import them to other resource markets that provide different service types.

In contrast, existing deployment and research efforts around promising ideas, specifically sidechains, fall short in exploiting their full potential~\cite{Back14,Garoffolo20,Gavzi19,Kiayias19}. They mostly focus on \emph{currency transfer} known as two-way peg (as we explain in Section~\ref{sec:background}), many present only high-level ideas, and do not include rigorous security models or analysis. Moreover, \emph{all} existing sidechain work has only dealt with sidechains that are \emph{independent} of the mainchain, which limits the use cases and performance gains that can be achieved. That is, those solutions cannot be trivially generalized to allow arbitrary data exchange and workload sharing due to this independence; each chain has its own transactions, miners, and tokens, making the delegation of operations from the mainchain to a sidechain, and synchronizing the result back, difficult. Although these were posed to be doable, no concrete sidechain protocol with security and performance analysis exists for that.

\textbf{An open question.} Therefore, we ask the following question: \emph{can we build a generic and secure efficiency solution for blockchain-based resource markets that has a unified architecture and interfaces, but allows for service-specific semantics?}

\subsection{Our Contributions} 
\vspace{-4pt}
We answer this question in the affirmative and propose \sysname, a secure performance booster for blockchain-based resource markets. At its core, \sysname employs a sidechain that has a (security and semantic) mutual-dependence with the market's main blockchain, to which all heavy/frequent service-related operations/transactions are offloaded. We make the following contributions. 

\textbf{System design.} We introduce a novel sidechain architecture that shares workload processing with the mainchain. It is composed of two types of blocks: \emph{temporary meta-blocks} and \emph{permanent summary-blocks}. The service-related workload in the market is handled by this sidechain, while the rest is handled by the mainchain. Instead of operating like a regular blockchain and storing all transactions, we devise a \emph{summarization and syncing protocol} and a \emph{block suppression} mechanism that the sidechain uses. That is, the sidechain records transactions into meta-blocks that are periodically summarized into smaller summary-blocks. The summary-blocks are used to sync the mainchain by updating the relevant state variables using sync-transactions, allowing for arbitrary data exchange between the two chains. Then, the meta-blocks whose summary is confirmed on the mainchain are pruned from the sidechain. This reduces the size of the sidechain as well as the mainchain as only concise summaries are stored permanently. Furthermore, this syncing process enables the mainchain to be the single truth of the system, thus simplifying system state tracking.

To expedite transaction processing, we employ a practical Byzantine fault tolerance (PBFT) consensus, with dynamic committees elected from the mainchain miners, to run the sidechain. Thus, once a transaction appears in a meta-block it is considered confirmed. \sysname is agnostic to the mainchain consensus protocol, and it does not assume a particular PBFT consensus or committee election mechanism. Its modular design permits using any secure protocols that realize these functionalities.

All of these design aspects form a unified architecture and interface of \sysname that can be used with any resource market regardless of the service type it provides. In turn, the summarization process is highly customizable allowing for expressive rules that can be tailored to suit the underlying service. These also promote flexibility as they support various modalities for service-payment exchange, market matching, etc., that can be configured based on the resource market design goals.

\textbf{Resilience and robustness techniques.} Due to the mutual-dependence relation between the main and side chains, interruptions on either of them may impact the other. To address that, we introduce techniques to achieve resilience and robustness. These include: (1) a mechanism to handle mainchain \emph{rollbacks} so the sidechain can tolerate changes in the recent history of the mainchain, and (2) an \emph{autorecovery protocol} to allow the sidechain to automatically recover from security threats related to malicious or unresponsive committees. The latter introduces the notion of \emph{backup committees} that will step in if any of these threats are detected to restore the valid and secure operation of the sidechain.

\textbf{Security analysis.} We analyze the security of our system showing that the deployment of \sysname on top of a secure resource market preserves its security in terms of service-related operations, and safety and liveness of the underlying mainchain. This is critical since we deal with dependent sidechains. As part of this analysis, we formally analyze the security of the autorecovery protocol showing bounds for its configuration parameters and the number of backup committees needed to ensure robustness and resilience against threats.

\textbf{Implementation and evaluation.} To assess efficiency, we build a proof-of-concept prototype of our scheme and conduct thorough benchmarks and experiments. In particular, as a use case, we show how \sysname can substantially improve the performance of a blockchain-based file storage market. Our experiments show that, for a market handling $\sim$2000 transactions per round, \sysname enhances throughput by $4-11$x and reduces its overall latency by $62-94$\%, depending on the sidechain configuration. Our results also show a reduction in the blockchain size by up to $\sim$90\%, and demonstrate support for a variety of server payment modalities at a low cost. We also compare our solution to an optimistic rollup-based solution, and show a $\sim$99.7\% reduction in finality.

Lastly, we note that, to the best of our knowledge, our work is the first to deal with dependent sidechains, and the first to show how sidechains can be used for blockchain pruning. Coupled with the performance gains achieved, we expect our work to advance the current state of the art and promote the practical viability of blockchain-based resource markets without compromising their security.

\section{Background}
\label{sec:background}
\vspace{-4pt}
We review the paradigm of blockchain-based resource markets and the building blocks used in our work, namely, sidechains and PBFT-based consensus.

\vspace{3pt}
\noindent\textbf{Blockchain-based resource markets.} Ethereum is among the earliest systems to provide a distributed service on top of the currency medium. It allows clients to deploy smart contracts and pay miners for the CPU cycles they spend when executing these contracts. Later models~\cite{filecoin,livepeer,golem} expanded the network to involve service-specific participant roles, e.g., servers. Instead of having all miners repeat the same task, as in Ethereum, servers are matched with clients to fulfill their service requests. 

\begin{figure}[t!]
\centerline{
\includegraphics[height= 1.4in, width = 0.85\columnwidth]{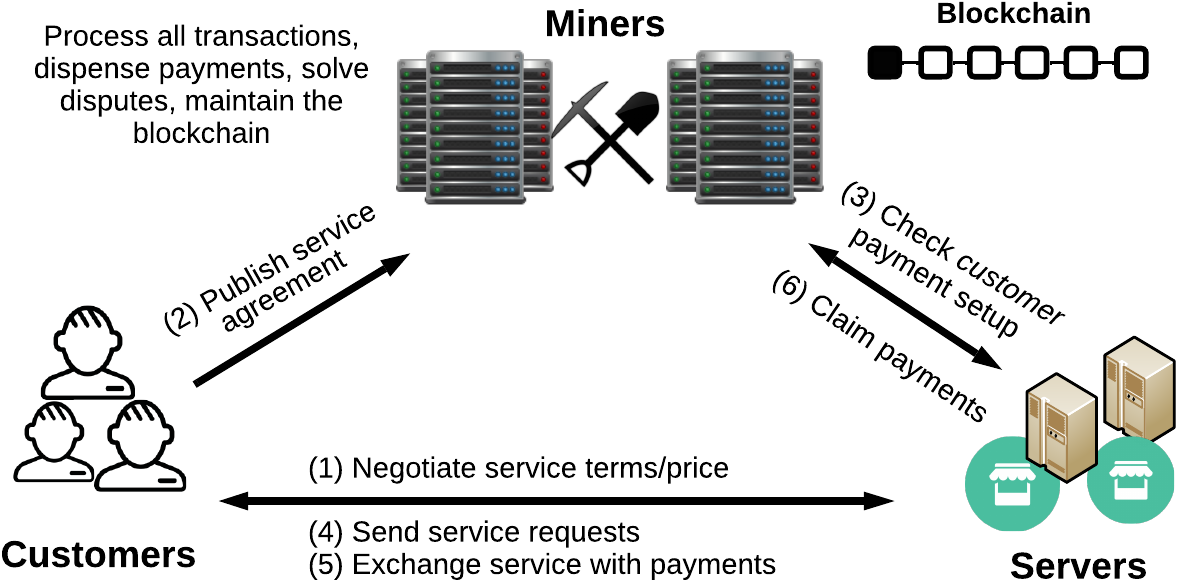}}
\vspace{-2pt}
\caption{A generic resource market model.} \label{resource-market}
\vspace{-8pt}
\end{figure}

A blockchain-based resource market consists of several modules (Figure~\ref{resource-market}): order matching and service term agreement, service-payment exchange, payment processing, and dispute resolution. After negotiating and committing to service terms, clients send requests to servers and pay in return. Service delivery proofs are needed to defend against colluding attackers who may collect payments without doing any work~\cite{Moran19,Fisch19}. Furthermore, due to the impossibility of fair exchange between untrusted parties~\cite{Even80}, additional countermeasures are required for service-payment exchange. For example, clients create escrows for the expected total amount, and miners dispense payments to the servers---out of the escrows---in return for valid service delivery proofs. 

Other solutions may divide the service into small amounts and does the same with payments (known as micropayments~\cite{Pass15}). This reduces financial risks: if a malicious server does not deliver a service, the client loses a small payment; and if a malicious client decides not to pay, the server loses payment for a small service amount. It also promotes flexibility as clients can stop the service at anytime, e.g., as in online gaming or pay-per-minute video streaming. Lastly, solving disputes involves processing cheating claims against participants and penalizing them.

\vspace{3pt}
\noindent\textbf{Sidechains.} A sidechain is a secondary blockchain tied to a mainchain. This term was first coined by Back et al.~\cite{Back14}, who also defined the main use cases of sidechains: supporting interoperability, enhancing scalability, and extending mainchain capability. That is, having a sidechain enables mainchain users to interact with the sidechain, allows for workload sharing, and provides an environment where new functionalities can be deployed/tested.

Despite the large body of work around sidechains (see Section~\ref{sec:related-work}), and that in theory the above use cases are viable, most existing research and deployment efforts focus only on currency exchange, known as two-way peg. That is, the whole goal is allowing clients to exchange two different cryptocurrency tokens between the two chains (call them token $A$ and token $B$). Usually, the mainchain miners are not aware of the sidechain, and the sidechain operation involves succinct proofs proving that some amount of token $A$ has been locked on the mainchain, so that an equivalent amount of token $B$ can be used on the sidechain. Two-way peg is an important operation that may suffice for mainchain interactions with an independent sidechain, the only sidechain type that existing solutions consider. This is exemplified in a firewall security property~\cite{Gavzi19} stating that if a sidechain gets compromised it will not impact the mainchain.

These existing models and notions do not suffice for our purposes; we require a sidechain that has a mutual-dependence relation with the mainchain allowing for arbitrary data exchange, not only currency, thus exploiting the full potential of sidechains. Again, although this is possible in theory, none of the existing research works or practical deployments demonstrate protocols for that, or how this can be used with resource markets.

\vspace{3pt}
\noindent\textbf{PBFT-based consensus.} To increase throughput and reduce confirmation delays, several consensus protocols utilize practical Byzantine fault tolerant (PBFT)~\cite{castro1999practical}. We distinguish two flavors: leader-based consensus, e.g.,~\cite{Kogias16}, and voting-based consensus, e.g.,~\cite{Gilad17}. In both cases, the assumption is that an attacker can corrupt less than one-third of the committee members running the agreement. Moreover, network operation is divided into epochs (an epoch is $k$ consecutive rounds and a round is the time period during which a new block is mined).

In leader-based consensus, the committee leader proposes a block in a round and collects signatures from the committee members indicating their agreement. In voting-based consensus, the committee is divided into block proposers and voters, where each proposer proposes a block that voters vote on (by signing). The block with vote majority will be added to the blockchain. Having a committee, where its size is much smaller than the whole miner population, reduces communication cost and fork probability, and speeds up agreement. Due to these advantages, we adopt a PBFT-based consensus to run the sidechain in \sysname.

Committee election is a crucial component. Obviously, having a static committee is insecure since it is vulnerable to targeted attacks. Thus, the dynamic approach is used in practice: in each epoch a new committee is elected. Since we deal with open-access blockchains, election relies on Sybil-resistant identities. For example, in proof-of-work blockchains, the probability of electing a miner is proportional to the computing power this miner owns (as in ByzCoin~\cite{Kogias16}). To address targeted attacks resulting from knowing the committee members in advance, cryptographic sortition~\cite{Gilad17} in the voting-based model has been introduced.  A party reveals that it was elected after proposing a block or voting (with a proof of being elected), rendering targeted attacks ineffective. Various committee election mechanisms with different trade-offs have been studied in the literature~\cite{Kogias16,pass2017hybrid,Gilad17,kiayias2017ouroboros}.

\sysname's design is agnostic to the consensus type of the underlying mainchain. Also, for its sidechain, \sysname does not assume a particular PBFT-based consensus or committee election algorithm. Any secure protocols that realize these functionalities can be used. To simplify the presentation, we use the leader-based approach when introducing our design.

\section{Preliminaries}
\label{sec:prelim}
%\vspace{-4pt}

\noindent\textbf{Notation.} We use $\lambda$ for the security parameter, and $\param$ for the system public parameters. Each participant maintains a secret and public key, $\sk$ and $\pk$, respectively. We use $\led$ to denote a ledger (or blockchain), $\led_{\mainc}$, and $\led_{\sidec}$ to denote the mainchain and sidechain ledgers, respectively. We use $\ppt$ as a shorthand for probabilistic polynomial time.   

\vspace{3pt}
\noindent\textbf{System model.} 
\sysname is designed to complement and enhance open-access distributed resource markets. Anyone can join, or leave, at any time, and these parties are known using their respective public keys. The resource market mainchain may run any consensus protocol and mining process, and miners establish Sybil-resistant identities based on their mining power. \sysname operates in rounds and epochs (as defined earlier).

A blockchain-based resource market is run by a set of clients $\clients$, servers $\servers$, and miners $\miners$.\footnote{In some models, servers play the role of miners and their mining power is represented by the amount of service they provide in the system.} It maintains a ledger $\led$ and provides the following functionalities:
\begin{description}
    \item[$\syssetup(1^{\lambda}) \rightarrow (\param, \led_0)$:] Takes as input the security parameter $\lambda$, and outputs the system public parameters $\param$ and an initial ledger state $\led_0$ (which is the genesis block).\footnote{For the rest of the algorithms, the input $\param$ is implicit.}

    \item[$\partysetup(\param) \rightarrow (\stt)$:] Takes as input the public parameters $\param$, and outputs the initial state of the party $\stt$ (which contains a keypair $(\sk, \pk)$, and in case of miners, the current view of the ledger $\led$).

    \item[$\createTransaction(\txtype, \aux) \rightarrow (\tx)$:] Takes as input the transaction type $\txtype$ and any additional information/inputs $\aux$, and outputs a transaction $\tx$ of one of the following types:
    \begin{itemize}
        \item $\tx_{\ask}$: Allows a client to state its service needs.
        \item $\tx_{\offer}$: Allows a server to state its service offering.
        \item $\tx_{\agreement}$: The service agreement between a client and a server (or a set of servers).
        \item $\tx_{\serviceProof}$: A service delivery proof submitted by a server.
        \item $\tx_{\servicePayment}$: Payment compensation from a client to a server for the service provided.
        \item $\tx_{\dispute}$: Initiates a dispute-solving process for a particular service misbehavior incident.
        \item $\tx_{\transfer}$: Currency transfer between participants.
    \end{itemize}

    \item[$\verifyTransaction(\tx) \rightarrow (0/1)$:] Takes as input a transaction $\tx$, and outputs 1 if $\tx$ is valid based on the syntax/semantics of its type, and 0 otherwise.

    \item[$\verifyBlock(\led_{\mainc}, \block) \rightarrow (0/1)$:] Takes as input the current ledger state $\led_{\mainc}$ and a new block $\block$, and outputs 1 if $\block$ is valid based on the syntax/semantics of blocks, and 0 otherwise.

    \item[$\updateState(\led_{\mainc}, \{\tx_i\}) \rightarrow (\led_{\mainc}')$:] Takes as input the current ledger state $\led_{\mainc}$, and a set of pending transactions $\{\tx_i\}$. It reflects the changes induced by these transactions and outputs a new ledger state $\led_{\mainc}'$.
\end{description}

Note that resource markets may offer non-service-related transactions other than $\tx_{\transfer}$. The notion above can be extended to cover these other types. Also, the set of transactions needed to operate the service can be extended; the ones in the notion above are intended to be a base set. Lastly, $\updateState$ is the process of mining a new block to reflect the changes induced by the executed set of transactions. Agreeing on this block and how it is mined is governed by the consensus protocol operating the mainchain of the resource market.

\begin{figure*}[ht!]
\centerline{
\includegraphics[height= 1.2in, width = 1.95\columnwidth]{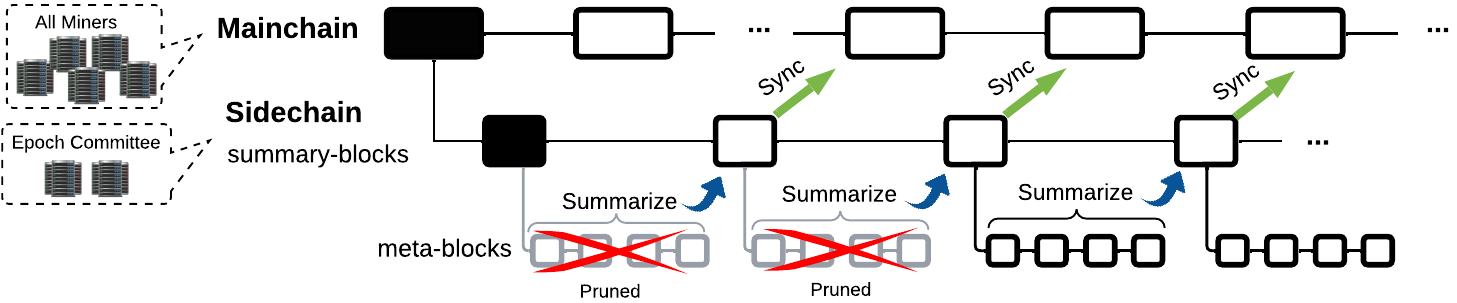}}
%\vspace{-4pt}
\caption{\sysname diagram.} 
\label{chainboost}
\vspace{-8pt}
\end{figure*}

As mentioned before, \sysname operates a sidechain managed by a committee elected from the mainchain miners who runs a PBFT-based consensus to mine new blocks. Furthermore, the sidechain contains two types of blocks: meta-blocks that record transactions, and summary-blocks that summarize meta-blocks mined in an epoch. Also, this committee issues a sync-transaction to sync the mainchain. We abstract these functionalities as follows.
\begin{description}
    \item[$\setup(\led_{\mainc}^0) \rightarrow (\led_{\sidec}^0, \led_{\mainc}')$:] Takes as input the mainchain genesis block $\led_{\mainc}^0$. It outputs the initial sidechain ledger state $\led_{\sidec}^0$ (which is the genesis block of the sidechain containing the sidechain public parameters $\param'$ that include the traffic classification and summary rules). Also, it creates the mainchain state variables that will be synced by the sidechain, thus producing a new mainchain ledger state $\led_{\mainc}'$.

    \item[$\elect(\led_{\mainc}) \rightarrow (\{\com_i\}, \{\leader_i\})$:] Takes as input the current state of the mainchain ledger $\led_{\mainc}$, and outputs a set of committees $\{\com_i\}$ and their leaders $\{\leader_i\}$.

    \item[$\createSyncTransaction(\aux) \rightarrow (\tx_{\sync})$:] Takes as input information $\aux$, and outputs a sync-transaction $\tx_{\sync}$.

    \item[$\verifySyncTransaction(\led_{\sidec},\tx_{\sync}) \rightarrow (0/1)$:] Takes as input the current sidechain ledger state $\led_{\sidec}$ and sync-transaction $\tx_{\sync}$, and outputs 1 if $\tx_{\sync}$ is valid based on its syntax/semantics, and 0 otherwise.

    \item[$\verifyBlock(\led_{\sidec}, \block_{\btype}) \rightarrow (0/1)$:] Takes as input the current sidechain ledger state $\led_{\sidec}$, a new block $\block$ with type $\btype = \meta$ or $\btype = \summary$. It outputs 1 if $\block$ is valid based on the syntax/semantics of the particular block type, and 0 otherwise.

    \item[$\updateState(\led_{\sidec}, \aux, \btype) \rightarrow (\led_{\sidec}')$:] Takes as input the current sidechain ledger state $\led_{\sidec}$, and a set of pending transactions $\aux = \{\tx_i\}$ (if $\btype = \meta$) or $\bot$ (if $\btype = \summary$ since the inputs are meta-blocks from $\led_{\sidec}$). It reflects the changes induced by $\aux$ and outputs a new ledger state $\led_{\sidec}'$.

    \item[$\prune(\led_{\sidec}) \rightarrow (\led_{\sidec}')$:] Takes as input the current sidechain ledger state $\led_{\sidec}$, and produces an updated state $\led_{\sidec}'$ in which all stale meta-blocks are dropped.
\end{description}

Note that $\elect$ may produce only one committee, the primary one that runs the sidechain, or a list of primary and backup committees as in \sysname's design.

\vspace{3pt}
\noindent\textbf{Security model.}
Our goal is to develop a secure efficiency solution---one that preserves the security of the underlying resource market (both its blockchain or ledger and the distributed service it provides). So, starting with a secure blockchain-based resource market, the security properties of this system must be preserved by \sysname. 

\emph{Ledger security.} A ledger $\led$ is secure if it satisfies the following properties~\cite{garay2015bitcoin} (the confirmed state of $\led$ includes all blocks buried under at least $k$ blocks, where $k$ is the depth parameter):
\begin{itemize}
\item Safety: For any two time rounds $t_1$ and $t_2$ such that $t_1 \leq t_2$, and any two honest parties $P_1$ and $P_2$, the confirmed state of $\led$ maintained by $P_1$ at $t_1$ is a prefix of the confirmed state of $\led$ maintained by party $P_2$ at time $t_2$ with overwhelming probability.

\item Liveness: If a transaction $tx$ is broadcast at time round $t$, then with overwhelming probability $tx$ will be recorded on within the confirmed state of $\led$ at time at most $t + u$, where $u$ is the liveness parameter.
\end{itemize}

Based on~\cite{pass2017analysis}, safety (or persistence) covers the consistency and future self-consistency properties, and liveness covers chain growth and quality properties.\footnote{Chain quality is defined as follows: in any sufficiently long list of blocks on $\led$, at least $t$ of these blocks are mined by honest miners where $t$ is a chain quality parameter.} A ledger protocol is parameterized by predicates to verify transaction and chain validity, ensuring that $\led$ records only valid transactions/blocks. For us, this also covers the validity of the additional operations/transactions introduced by the resource market so that only valid ones are accepted.

\emph{On the security of distributed resource markets.} This has been thoroughly studied in~\cite{Almashaqbeh19}, where that work defined threat categories related to the service and the blockchain/cryptocurrency itself. The former mainly includes denial of service, service corruption, service slacking (servers collect payments without all promised service work), and service theft (clients obtain service without paying the servers the amount they agreed on). While the latter covers violating the security properties of the blockchain defined above. A secure resource market means one that protects against all these threats.

The goal of \sysname is to optimize the performance of a given resource market, in terms of throughput, confirmation delays, and blockchain size, using a dependent sidechain. The way the offloaded workload to this sidechain is processed is identical to the logic used by the resource market. Thus, in our security analysis, we show that \sysname and all the techniques it introduces preserve the security of the underlying resource market. 

\vspace{3pt}
\noindent\textbf{Adversary model.}
\sysname runs on top of a secure resource market that operates a secure blockchain as defined previously. We distinguish three miner behaviors; honest who follow the protocol as prescribed, malicious (controlled by an adversary) who may deviate arbitrarily, and honest-but-lazy who collaborates passively with the adversary by accepting records without validation or going unresponsive during consensus.\footnote{We consider lazy miners on the sidechain to account for scenarios where miners are well-motivated to work honestly on the mainchain. But for the sidechain, with its heavy load, they might choose to be lazy in order to use all their resources for providing service or for the mainchain mining to collect more payments/rewards, i.e., the verifier dilemma~\cite{luu2015demystifying}.} The adversary can introduce new nodes or corrupt existing ones, without going above the threshold of faulty nodes of the consensus protocol. The adversary can see all messages and transactions sent in the system (since we deal with public open-access blockchains) and can decide their strategy based on that. This adversary can reorder messages and delay them; we assume bounded-delay message delivery, so any sent message (or transaction) will be delivered within $\Delta$ time as in~\cite{Gilad17,pass2017analysis,Kokoris18}. We assume slowly-adaptive adversaries~\cite{avarikioti2019divide} that can corrupt parties (specifically miners) only at the epoch beginning. Lastly, we deal with $\ppt$ adversaries who cannot break the security of the used cryptographic primitives with non-negligible probability.

\section{\sysname Design}
\label{sec:design}
In this section, we present \sysname starting with an overview of its design, followed by an elaboration on the technical details, limitations, and security.

\subsection{Overview}
\label{sec:overview}
\vspace{-3pt}
\sysname introduces a new sidechain architecture and a set of techniques to securely boost the performance of blockchain-based resource markets. This sidechain has a mutual-dependence relationship with the mainchain as they share the market workload. In particular, all heavy/frequent service-related operations/transactions are processed by the sidechain, while only brief summaries of the resulting state changes are logged on the mainchain.

As shown in Figure~\ref{chainboost} (and Figure~\ref{fig:comparison} shows a detailed view \sysname's impact on processing and storage), the sidechain works in parallel to the mainchain and shares its workload processing and storage. The transactions in the system are classified into sidechain and mainchain transactions. All service-related operations that can be summarized (such as service delivery proofs, service payments, and dispute-solving-related transactions) will go to the sidechain, while the rest will stay on the mainchain. Thus, in the setup phase, system designers determine this classification and the summarization rules (in the next section, we provide guidelines based on the generic resource market paradigm presented earlier). The sidechain is launched with respect to the mainchain genesis block. Both chains operate in epochs and rounds, where the epoch's length determines the frequency of the summaries.

During each epoch, a committee from the mainchain miners is elected to manage the sidechain. The rest of the miners do not track the traffic processed by the sidechain, thus reducing their workload. The sidechain is composed of \emph{two types of blocks}: temporary meta-blocks and permanent summary-blocks. For each sidechain round, the committee runs a PBFT-based consensus to produce a meta-block containing the transactions they processed. In the last round of the epoch, this committee produces a summary-block summarizing all state changes imposed by the meta-blocks within that epoch.

To maintain a single truth of the market, represented by its mainchain, \sysname introduces a periodic \emph{syncing process}. Once a summary-block is published, the committee issues a sync-transaction including the summarized state changes. The mainchain miners process this transaction by updating the relevant state variables on the mainchain. To reduce the storage footprint of the sidechain, and subsequently the mainchain size, \sysname introduces a \emph{pruning mechanism}; when the sync-transaction is confirmed on the mainchain, all meta-blocks used to produce the respective summary-block are discarded.

The syncing and pruning processes have positive impacts on market operation. Maintaining a single source of truth (i.e., the mainchain) simplifies tracking the system state. At the same time, having permanent summary-blocks supports public verifiability as anyone can verify the source of the state changes on the mainchain using the summaries and any other data (that might be published in full in the summary-blocks). Moreover, the pruning process reduces the storage footprint, thus allowing for faster bootstrapping of new miners.

Mutually dependent chains are challenging; if one fails so will the other. Also, network propagation delays and connectivity issues may lead to temporary forks. A question that arises here is how to securely handle rollback situations in \sysname---when a sidechain has already synced with the mainchain but subsequently the recent mainchain blocks are rolled back (i.e., they belong to an abandoned branch of the mainchain). Another question is related to how active or passive attacks, or interruptions, on the sidechain may impact the mainchain. \sysname introduces a \emph{mass syncing} mechanism and an \emph{autorecovery protocol} to handle rollbacks, enable a sidechain to recover after periods of interruption, and allow the mainchain to tolerate these interruptions. We provide more details in Section~\ref{sec:robustness}.

\begin{figure}[t!]
	%\vspace{-12pt}
    \centering

    \subfloat[Without \sysname]{%
  \includegraphics[clip,width=\columnwidth]{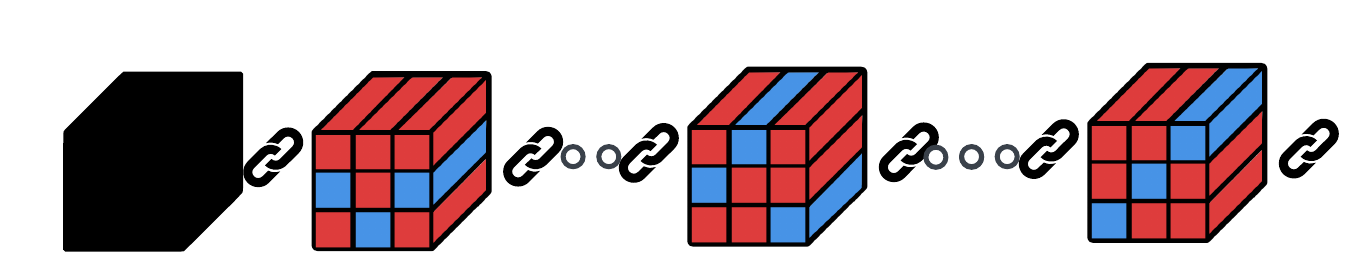}%
%\vspace{-4pt}
}

\subfloat[With \sysname]{%
  \includegraphics[clip,width=\columnwidth]{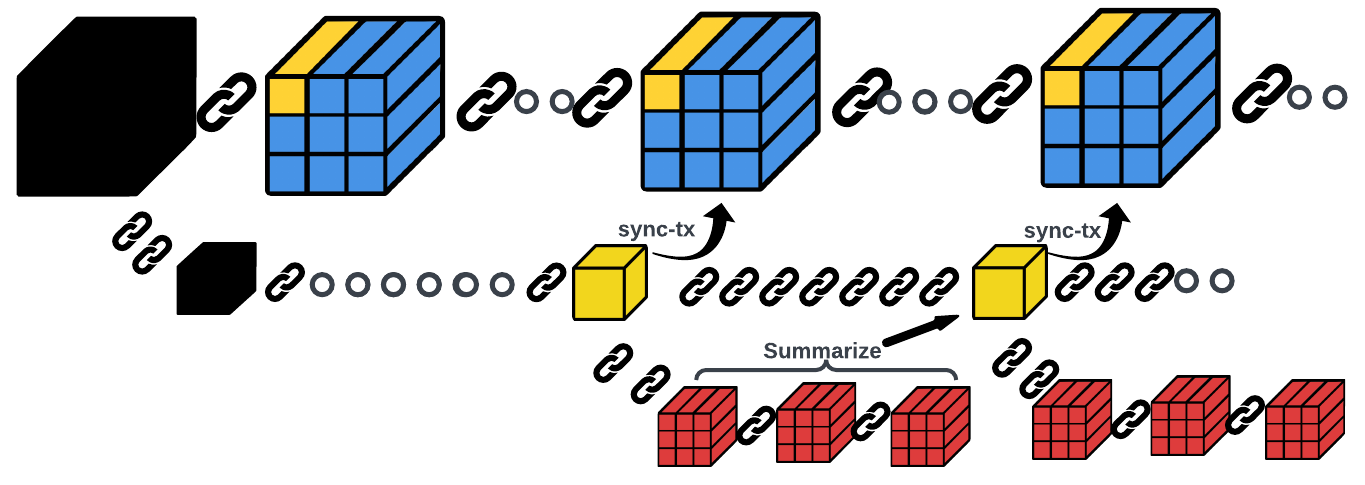}%
  %\vspace{-4pt}
  }
  %\vspace{-3pt}
    \caption{Mainchain with and without \sysname (service transactions are in red, others are in blue. Summary-blocks and sync-transactions are in yellow).}
    \label{fig:comparison}
    %\vspace{-6pt}
\end{figure}

As noted, the design above is generic and not tied to a particular resource market type. It resembles a unified architecture and interface between the main and side chains that any resource market can use. Capturing service-specific semantics is accounted for in the setup of the traffic classification and summary rules. That is, based on full knowledge of the market configuration and operation, service-related transactions are specified and will be handled by the sidechain. The same applies to how this traffic should be summarized; these will generally follow a generic paradigm (that we discuss later) that can be adapted to fit the underlying market. Furthermore, \sysname targets markets with substantial service-related traffic (we show that performance gains are minimal in case most traffic is non-service related and must reside on the mainchain). Moreover, the sidechain configuration, in terms of round and epoch duration, and block sizes, impact the achieved performance gains.  

\textbf{Concrete examples.} As a concrete example of the benefit of \sysname, we look at a file storage network (or any other resource market that involves periodic service delivery proofs). A server periodically engages in a proof-of-storage protocol to prove it stores a client's file. If the proof is valid, the client pays the server for the service provided. Those proof-of-storage transactions are usually plentiful, which slows down the processing of other transactions in the system and consumes space on the blockchain. Using \sysname, the proofs are processed by the sidechain. When generating a summary-block, the committee tallies the proofs generated by each server, and then the mainchain is synced. Thus, the mainchain maintains only a state table of proof counts for active servers. Finally, meta-blocks containing old proof-of-storage transactions whose summary got confirmed on the mainchain are pruned. This greatly reduces the chain size and speeds up transaction confirmation.

Another way \sysname can be beneficial is for the market matching module. In any resource market, clients and servers publish their service demands (i.e., asks) and offers, respectively, and get matched after negotiating the service terms. With \sysname, those demands and offers are recorded on the sidechain. The summary-blocks will only contain finalized service agreements, that get synced back up with the mainchain. This is a viable way to move the service matching from an out-of-band process into a component of the market, enhancing accountability, trust, and transparency. This is done without increasing the burden on the mainchain.

\subsection{Technical Details}
Now we delve into the technical aspects of \sysname's setup phase, operation, architecture, and robustness.

\subsubsection{Setup Phase}
Introducing \sysname changes the way the underlying resource market behaves. Transactions are separated into two categories based on their home chain. Also, a subset of the mainchain miners will be responsible of managing the sidechain on a rotating basis.

\vspace{3pt}
\noindent{\bf Configuration, traffic classification, and summary rules.} 
\sysname builds a unified architecture and interface that can be used for any resource market. At the same time, it permits tailored customization that suits the service type the market provides, which are set during the setup phase (Figure~\ref{fig:syssetup}). \sysname is most beneficial for resource markets with substantial service-related traffic. For configuring the public parameters $\param'$ of the sidechain, larger sidechain block sizes, with short round duration, would be more effective. However, this size should not be larger that typical recommended sizes in practice (e.g., up to 2 MB) and the duration must be enough for the sidechain committee to agree on a block.

\begin{figure}[t!]
\begin{framed}

 $\setup(\led_{\mainc}^0)$: Takes as input the mainchain genesis block (which includes the public parameters $\param$). Generates the sidechain configuration parameters:
 \begin{itemize}
     \item The epoch length $\omega$.   
     \item All public parameters $\param'$ needed by the sidechain protocol.
     \item Traffic classification.
     \item Summary rules.
     \item Summary state variables for the mainchain $\mathsf{summary_{variables}}$.
 \end{itemize}
 
 \vspace{3pt}
 Updates $\led_{\mainc}$ to produce a new state $\led_{\mainc}'$ to include $\param'$ and $\mathsf{summary_{variables}}$.
 
 \vspace{5pt}
 Initializes:
 \begin{itemize}     
     \item $\mathsf{pendingTx_{\mainc}}$, pending transaction pool to keep track of pending mainchain traffic.
     \item  $\mathsf{pendingTx_{\sidec}}$, pending transaction pool to keep track of pending sidechain traffic.
 \end{itemize}

\vspace{3pt}
 Outputs: epoch length $\omega$, sidechain genesis block $\led_{\sidec}^0$ (that references $\led_{\mainc}^0$), and $\led_{\mainc}'$.
 
\end{framed}
\vspace{-6pt}
\caption{System setup.} 
\vspace{-6pt}
\label{fig:syssetup}
\end{figure}

\sysname provides a general framework to distribute the transactions between the mainchain and the sidechain, ensuring effective load sharing. That is, service-related transactions go into the sidechain. They are usually frequent and can be summarized---in terms of the abstract notion we defined these include $\tx_\ask$, $\tx_\offer$, $\tx_\agreement$, $\tx_\servicePayment$, $\tx_\serviceProof$, and $\tx_{\dispute}$. The mainchain processes and records non-service related transactions, e.g. minting new coins and currency transfers, or those that cannot be summarized, e.g., escrow creation. \sysname also provides a flexible framework for summarizing service transactions. For example, service delivery proofs $\tx_\serviceProof$ can be summarized by simply counting them, and this is all that the mainchain needs to track, and the same applies for the accumulated service payments $\tx_\servicePayment$. For solving disputes, the proof-of-cheating transaction and the dispute outcome---validity and the imposed punishment (if any)---resemble the incident summary. We formalize these summary rules in a generic way (based on our abstract model) in Figure~\ref{fig:summary-rules}. Those rules are not fixed; they can be adapted/extended to fit the underlying resource market specifications.

We propose a simple approach that relies on static classification and annotation. After deciding the workload split, \sysname adds an extra field to the transaction header indicating its type, like 00 for mainchain transactions, and 11 for sidechain transactions. This annotation is set in the network protocol and examined by the miners to determine the home chain of each transaction they receive. The summarization rules are encoded into the sidechain protocol, to be used by the committee to produce succinct state changes of the traffic they process. Those rules can be changed at later stages as long as the interpretation of summaries does not change, allowing for greater flexibility while preserving compatibility. For example, a different proof of service delivery construction can be used while the summary remains as the proof count.

Tying that to our abstract model, the transactions $\tx_{\ask}$, $\tx_{\offer}$, $\tx_\agreement$, $\tx_{\serviceProof}$, $\tx_{\servicePayment}$, and $\tx_{\dispute}$ will all be annotated with 11 (i.e., sidechain traffic), while the rest will be annotated with 00 as they should reside on the mainchain. Also, the mainchain now processes a new transaction type, namely, the sync-transaction $\tx_{\sync}$. The outcome of the sidechain setup phase, as shown in Figure~\ref{fig:syssetup}, includes an updated state of the mainchain ledger to record the state variables used for the summary and all sidechain configuration parameters, in addition to the epoch length $\omega$. Furthermore, as the sidechain is tied to the mainchain, its genesis block will reference the mainchain's genesis block.

\begin{figure}[t!]
\begin{framed}

Input: meta-blocks $\block^1_\meta, \dots, \block^n_\meta$ from an epoch.

\vspace{3pt}
Initialize empty summary structures $\summ_\serviceProof$,  $\summ_\servicePayment$, $\summ_\dispute$, and $\summ_\agreement$.

\vspace{3pt}
\forloop{$i \in \{1, \dots, n\}$ and every $\tx \in \block^i_\meta$}

\;\;\;\;\ifcond{ $\tx.\txtype = \tx_{\serviceProof}$ }

\;\;\;\;\; // $\cid$ is the service contract ID

\;\;\;\;\;\; ++ $\summ_\serviceProof[\tx.\cid]$

\;\;\;\;\elsecond \ifcond{ $\tx.\txtype = \tx_{\servicePayment}$ }

\;\;\;\;\;\; $\summ_\servicePayment[\tx.\cid]$ += $\tx.\amount$

\;\;\;\;\elsecond \ifcond{ $\tx.\txtype = \tx_{\dispute}$ }

\;\;\;\;\;\; $\summ_\dispute[\tx.\cid] = (\tx\mathsf{.proof}, \tx\mathsf{.outcome})$ 

\;\;\;\; // $\tx_{\agreement}$ references $\tx_{\ask}$ and $\tx_{\offer}$, so the 

\;\;\;\;\;// summary covers these as well ($\mathsf{s}$ is the server 

\;\;\;\;\;// and $\mathsf{cl}$ is the client)

\;\;\;\;\elsecond \ifcond{ $\tx.\txtype =\tx_{\agreement}$ }

\;\;\;\;\;\; $\summ_\agreement[\tx.\cid]$ = $(\tx\mathsf{.s}, \tx\mathsf{.cl},\tx\mathsf{.terms})$

Output $\summ_\serviceProof$,  $\summ_\servicePayment$, $\summ_\dispute$, and $\summ_\agreement$
\end{framed}
\vspace{-6pt}
\caption{Summary rules (assuming one server per a service contract. If it is a set of servers instead, then indexing should indicate the particular server using $\tx.\cid.\mathsf{s}$).} 
\vspace{-6pt}
\label{fig:summary-rules}
\end{figure}

\vspace{3pt}
\noindent{\bf Managing the sidechain.} Having all miners participate in consensus on both chains is counterproductive. Furthermore, as the goal is to speed up service delivery by reducing transaction confirmation time, the sidechain requires a fast consensus protocol. Therefore, we let the sidechain be managed by a committee elected from the mainchain miners on a rotating basis, and this committee runs a PBFT-based consensus. Although we favor voting-based PBFT to mitigate targeted attacks, to simplify the discussion, we use the leader-based one. Nonetheless, \sysname can be used with either type.

At the onset of each epoch, a new committee is elected to manage the sidechain, which distributes the sidechain workload handling among all miners (i.e., each miner will participate in this task when elected). Any secure committee election can be used, e.g.,~\cite{Gilad17} (as we discuss in Section~\ref{sec:background}). To enable fast handover between committees, all miners in the network maintain a copy of the sidechain.

\subsubsection{Architecture and Operation} 
\label{sub:arch}
When integrated with a resource market, \sysname must preserve the security (safety and liveness) of the market's blockchain and the security/validity of its decentralized service. That is, the system state must be valid, publicly verifiable and available, immutable, and growing over time. We devise a new architecture that allows \sysname to preserve those properties while reducing the overhead of service-related transactions. The operation of \sysname's sidechain is depicted in Figure~\ref{fig:sysoperation}.

\vspace{3pt}
\noindent\textbf{Sidechain blocks and growth.} \sysname's sidechain is composed of temporary meta-blocks and permanent summary-blocks. During each round in an epoch, the committee leader proposes a meta-block, containing sidechain transactions processed in that round, and initiates an agreement to collect votes from the committee. Once enough votes are collected, the block is confirmed and is added to the sidechain. Validating a meta-block proceeds as in PBFT-based consensus; it encompasses verifying block correctness (based on the same verification rules used by the underlying resource market to process service-related transactions), and that valid votes by the legitimate members are provided. 

\begin{figure}[t!]
\begin{framed}
// All miners inspect incoming transactions annotation

\textbf{if} $\mathsf{Annot}(\tx) = 00$ \textbf{then}

\;\;\;\; \textbf{if} $\verifyTransaction(\tx)$ \textbf{then}

\;\;\;\;\;\; Add $\tx$ to $\mathsf{pendingTx_{\mainc}}$

// Done only by sidechain miners

\textbf{elseif} $\verifyTransaction(\tx)$ \textbf{then}

\;\;\;\; Add $\tx$ to $\mathsf{pendingTx_{\sidec}}$
\\\\
// Sidechain miners

// Check if it is the last round in the epoch

\textbf{if} $\mathsf{lastRound}$ \textbf{then}

\;\;\;\; $\led_{\sidec}' = \updateState(\led_{\sidec}, \bot, \mathsf{summary}$)

\;\;\;\; $\block_{\mathsf{summary}} = \led_{\sidec}'.\head$

\;\;\;\; $\tx_{\sync}$ = $\createSyncTransaction(\block_{\mathsf{summary}})$

\textbf{else}

\;\;\;\; $\led_{\sidec}' = \updateState(\led_{\sidec}, \mathsf{pendingTx_{\sidec}, meta}$)
\\\\
// Mainchain miners

\textbf{if} $\verifySyncTransaction(\led_{\sidec},\tx_{\sync})$ \textbf{then}

\;\;\;\; Add $\tx_{\sync}$ to the head of $\mathsf{pendingTx_{\mainc}}$
\\\\
// When mining the next mainchain block, update 

// $\mathsf{variables_{summary}}$ based on $\tx_{\sync}$

$\led_{\mainc}' = \updateState(\led_{\mainc}, \mathsf{pendingTx}_{\mainc})$
\\\\
// When an epoch ends, elect new committees

$(\{\com_i\}, \{\mathsf{leader}_i\}) \leftarrow \elect(\led_{\mainc})$
\\\\
// When $\tx_{\sync}$ is confirmed on the mainchain,

// sidechain miners will prune the sidechain

$\led_{\sidec}' = \prune(\led_{\sidec})$
\end{framed}
\vspace{-6pt}
\caption{System operation.} 
\vspace{-6pt}
\label{fig:sysoperation}
\end{figure}

At the end of the epoch, the leader proposes a summary-block summarizing all meta-blocks in the epoch (based on the summary rules set during the setup phase). Similarly, a summary-block needs the committee agreement to be published on the sidechain. Also, validating this block involves validating the votes and verifying that the summaries are correct based on the summary rules and the content of the corresponding epoch meta-blocks.   

Since \sysname's sidechain has two block types, the way blocks refer to their predecessors, to preserve immutability, is different from regular blockchains. As shown in Figure~\ref{chainboost}, the sidechain genesis block references the mainchain genesis block, the first meta-block in an epoch references the summary-block of the previous epoch, while any other meta-block references its predecessor meta-block in that epoch. Finally, a new summary-block references the last epoch's summary-block.

\vspace{3pt}
\noindent{\bf The syncing process.} The sidechain committee uses the summary-blocks to sync the mainchain without requiring the mainchain miners to re-execute the offloaded operations. Every time a summary-block is published, the leader issues a sync-transaction $\tx_{\sync}$ containing the state changes. Mainchain miners will add this transaction to the head of the pending (mainchain) transaction pool, and will be used to update $\mathsf{variables_{summary}}$ on the mainchain. This transaction will be accepted if it is valid based on the referenced summary-block (recall that all miners have copies of both the mainchain and the sidechain).

\vspace{3pt}
\noindent{\bf Pruning stale blocks.} We introduce a \emph{block suppression} technique to prune the sidechain without compromising its immutability. \sysname can safely drop the meta-blocks once their corresponding $\tx_\sync$ is confirmed on the mainchain. This is done to ensure robust transition between epochs; when $\tx_\sync$ is published in a block buried under at least $k$ blocks on the mainchain, the corresponding meta-blocks can be dropped (so they are not pruned within one epoch period). Given that the summary-blocks are permanent, and that they represent valid summaries produced by a committee with an honest majority (we derive a lower bound for the committee size to guarantee that in Appendix~\ref{appendix:autorecovery-sec-analysis}), public verifiability of the mainchain is not impacted. Anyone can refer to the summary-blocks to verify relevant mainchain state changes.

\subsubsection{Robustness and Resilience}
\label{sec:robustness}
\sysname builds a sidechain that has a mutual-dependence relation with the mainchain. Thus, the security and validity of both chains are tied to each other. To maintain system operation, interruptions on any of these chains must be tolerated by the other. This encompasses handling: (1) block rollbacks on the mainchain, and (2) sidechain downtime from malicious attackers or even lazy honest parties. We devise mechanisms to address these cases.

\vspace{3pt}
\noindent{\bf Handling rollbacks.} Rollbacks happen when mainchain miners switch to a longer branch, causing the most recent blocks to become obsolete.\footnote{Since we assume a secure mainchain, miners may disagree only on the recent history of the blockchain.} This impacts \sysname's sidechain operation on two fronts: first, the committee election if the used mechanism relies on the view of the mainchain, and second, the syncing process if the abandoned blocks contain sync-transactions. 

For the first case, we restrict the view used for committee election to the confirmed blocks on the mainchain. For example, for the sliding window-based approach~\cite{Kogias16}, the committee is composed of the miners of the last $y$ confirmed blocks, and for stake-based election~\cite{Gilad17}, the amount of stake owned by each miner is computed based on the confirmed mainchain history. By doing so, even if a rollback happens, it will not change the sidechain committees elected during the rollback period.

For the second case, we introduce a \emph{mass-syncing} technique where a $\tx_\sync$ captures $x$ summary-blocks instead of one. If a rollback is detected, the current sidechain committee issues a $\tx_\sync$ based on the summary-block not only from the current epoch, but also from previous epochs impacted by the rollback. We set this rollback period to be the time needed to confirm a mainchain block.

\vspace{3pt}
\noindent{\bf Autorecovery protocol.} 
We introduce an autorecovery protocol to detect and recover from any sidechain interruptions.\footnote{Some forms of autorecovery can be found in the literature. For example, in Algorand~\cite{Gilad17}, miners mine empty blocks in the rounds when no agreement is reached to preserve liveness.  However,~\cite{badertscher2020consensus} showed that under a small adversarial spike (over the honest majority assumption) Algorand's network fails to self-heal indefinitely. In Spacemesh~\cite{spacemesh}, self-healing is discussed based on a mesh (layer DAG) blockchain, whereas \sysname does not restrict the mainchain structure.} We provide an overview of this protocol in this section, while a detailed threat modeling and protocol description can be found in Appendix~\ref{appendix:autorecovery}.

\begin{figure}[t!]
    \includegraphics[height= 1.5in, width = \columnwidth]{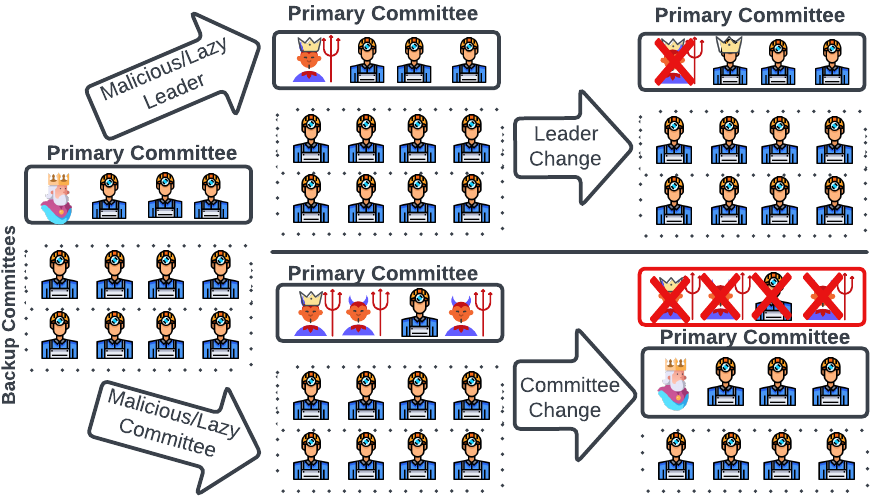}
    \caption{Autorecovery cases and mechanisms}
    \label{fig:autorecovery}
    \vspace{-6pt}
\end{figure}

\emph{Interruption cases.} To be on the conservative side, our threat model considers that a sidechain committee may contain honest miners, malicious miners (their number is denoted as $m$), and lazy (honest) ones (their number is denoted as $l$). As we explain in Appendix~\ref{appendix:autorecovery}, having $m+l \geq 2f+2$ (where $f$ is the maximum number of faulty nodes in the committee) violates safety---the committee might agree on invalid blocks, and having $m+l \geq f+1$ violates liveness---no agreement can be reached. Furthermore, a malicious leader may propose invalid blocks/$\tx_{\sync}$ (note that an invalid $\tx_{\sync}$ can be detected it is verified), and a lazy/unresponsive one may not propose a block or issue a $\tx_{\sync}$. These interruption cases are shown in Figure~\ref{fig:autorecovery}, and our protocol involves two techniques to detect and recover from them.

\textbf{(1)} \emph{Leader change.} We use the view-change technique used in PBFT protocols~\cite{castro1999practical} parameterized by a timeout $\zeta$. Each committee member monitors the behavior of the leader. If an invalid block/$\tx_{\sync}$ is proposed, each member sends a leader-change message. Once an agreement is reached, the next leader candidate will take over (thus, the committee election algorithm not only elects the leader but also the candidates for leader change). The outcome of the process, i.e., the agreed-upon leader-change message, will be published in a meta-block (and consequently in a summary-block) to announce the leader change. 

Furthermore, and to account for the case that committee members themselves could be unresponsive, and hence, no agreement on the new leader is reached, another timeout parameter $\eta$ is used, such that $\eta > \zeta$. If the leader change does not conclude in $\eta$ period, the committee is identified as unresponsive and will be handled by the backup committee technique (discussed below). Moreover, $\eta$ is also used to deal with the case when leader candidates are lazy/malicious, leading to a long sequence of leader changes---putting the committee on hold.

\textbf{(2)} \emph{Backup committees.} To deal with an unresponsive committee or one that may publish invalid blocks, we introduce the concept of backup committees. A backup committee has the same size as the primary committee, and it includes miners elected based on the older history of the mainchain, excluding the primary members. 

This committee monitors the primary committee's behavior and maintain all valid sidechain transactions they receive. When the backup committee detects an invalid block being published on the sidechain, or unresponsive primary committee (i.e., no activity during a timeout $\eta$), it initiates the process of taking over. The backup committee is subject to the same threat model of the primary committee, and could be susceptible to the same threat cases discussed above. Thus, we designate $\kappa$ backup committees, denoted as $\{\com_i\}_{i = 1}^{\kappa}$, each one monitors the one ahead of it in the line and takes over in case of misbehavior. Our protocol is captured in Figure~\ref{fig:autorecovery-pro}.

As shown in the figure, the monitoring period for each backup committee is larger than the one used by the committee ahead of it. That is, the first backup committee $\com_1$ is watching the primacy committee $\com_0$, while other backup committees will step in if the previous backup committee(s) and the primary committee are unresponsive/misbehaving. When the backup committee becomes the primary committee, it resumes operation (after discarding invalid blocks, if any). Moreover, in case syncing did not take place at the end of the epoch, but there is a valid summary-block, the next epoch committee will issue a mass sync transaction for its current epoch and the previous one. Full details can be found in Appendix~\ref{appendix:autorecovery}.

Such autorecovery capability cannot be unconditional; the security is guaranteed under an appropriate protocol configuration (such as the number of backup committees and the committee size). We formally analyze that showing how to configure these parameters in a way that makes the failure probability of our autorecovery protocol negligible (full analysis can be found in Appendix~\ref{appendix:autorecovery-sec-analysis}). 

\begin{figure}[t!]
\begin{framed}
At the beginning of any round in any epoch, each $\com_i$ does the following:
\vspace{3pt}

\textbf{case 1}---lack of progress:

\;\; Set $\led = \led_{\sidec}$

\;\; \textbf{if} $\led = \led_{\sidec}$ at the end of $i \cdot \eta$ duration \textbf{then} 
    
    \;\;\;\; Each member in $\com_i$ issues an unresponsive-committee message. 
\vspace{3pt}

\textbf{case 2}---invalid operation:

\;\; Receive a new $\block_{\mathsf{meta}}$ or $\block_{\mathsf{summary}}$.

\;\; // $\btype$ is either $\meta$ or $\summary$

\;\; \textbf{if} $\verifyBlock(\led_{\sidec},\block_{\btype}) = 0$  \textbf{then}

\;\;\;\; Each member in $\com_i$ issues a misbehaving-committee message.
\vspace{3pt}

\textbf{case 3}---no syncing:

\textbf{if} no $\tx_{\sync}$ in $\led_{\mainc}$ at the end of the epoch  \textbf{then}

\;\;\;\; Resort to mass-sync in the next epoch.

\vspace{6pt}
// Committee take over

For case 1 and 2, once an agreement is reached: 

\textbf{for} $i \in \{1, \dots, \kappa\}$
\;\; Set $\com_{i-1} = \com_{i}$
\end{framed}
\vspace{-6pt}
\caption{Autorecovery protocol---backup committees.} 
\vspace{-6pt}
\label{fig:autorecovery-pro}
\end{figure}

\subsubsection{Limitations and potential extensions}
\label{sec:limit}
We discuss limitations and potential extensions to our system. These are related to the applicability of \sysname, the overhead of its autorecovery protocol, the miner incentives to maintain both chains, and the issue of transaction fees based on the temporary storage aspect of the meta-blocks. We discuss these issues in Appendix~\ref{appx:limit}.

\subsection{Security}
As \sysname introduces a dependent-sidechain that shares the workload with the mainchain, and a block pruning mechanism, we have to show that, under this new architecture, the security of the resource market is preserved. In Appendix~\ref{appx:chainboost-security}, we prove the following theorem. 

\begin{theorem}
    \sysname preserves the safety and liveness (cf. Section~\ref{sec:prelim}) of the underlying resource market.
\end{theorem}

\section{Implementation}
\label{sec:implementation}
To assess the performance gains of \sysname, we implement a proof-of-concept and conduct experiments testing several impactful factors. In doing so, we chose a file storage resource market as a use case inspired by many real-world Web 3.0 systems building such service, e.g., Storj~\cite{website:storj}, Sia~\cite{website:sia}, and Filecoin~\cite{website:filecoin}. We discuss the implementation details in this section, and the experiment setup and evaluation results in the next section. Our implementation can be found at~\cite{chainboost-repo}.

\textbf{Sidechain implementation.}
We implemented our sidechain architecture in Go, including meta-blocks, summary-blocks, the summary rules, the syncing process, chain pruning, and the transaction annotation indicating the traffic split between the chains. In our implementation, each miner maintains a copy of the mainchain and the sidechain. Miners on the sidechain committee maintain a separate queue for each chain pending transactions (mainchain miners follow only the mainchain traffic).

For the PBFT protocol run by the sidechain, we use the collective signing (CoSi)-based PBFT-based consensus using BLSCoSi from Cothority~\cite{github:cothority}. For the communication between miners, we employ the Cothority Overlay Network Library (Onet)~\cite{github:onet}. For committee election, we use the sliding window mechanism from~\cite{Kogias16}. We chose this simple approach as it suffices to show the performance gains that \sysname can achieve. In the bootstrapping phase of the sliding window, our implementation picks the members of the committee at random. Later on, each potential committee member is added to the committee once, even if it mined multiple mainchain blocks within the sliding window. Finally, the most recently added committee member will be the committee leader for the epoch. 

Each sidechain round proceeds as follows. The leader creates a block (a meta-block, or if it is the epoch end, a summary-block) and initiates an agreement using the CoSi-based PBFT. When an agreement is reached, the leader publishes the block, along with the aggregated committee signature produced via the CoSi process, on the sidechain. At the end of the epoch, the leader issues a sync-transaction, based on the summary-block, and broadcasts it to the mainchain miners. This transaction has the highest priority, thus it will be placed first in the mainchain pending transaction queue, for fast processing.

\textbf{Use case: file storage resource market.}
We use a decentralized file storage network (inspired by Filecoin~\cite{filecoin}) as a use case exemplifying a resource market to show the benefits that \sysname can achieve. In this market, there are two types of participants: the storage providers or servers (who also serve as miners since the amount of storage service they provide represents their mining power), and clients who want to use the service. In our implementation, this market is composed of two modules:

\emph{Market matching module}: A client asks for service for a specific duration, file size and price using a propose-contract transaction (which also creates an escrow to pay for the service). The server answers with a commit-contract transaction matching a specific client offer. Both of these transactions are published on the mainchain. In our implementation, offer and demand generation is automatically done, and whenever a contract expires, a new set of propose and commit transactions are issued.

\emph{Service-payment exchange module}: A server proves that it is storing a file by issuing a proof-of-storage every mainchain round. In our proof-of-concept, we use compact (non-interactive) proof-of-retrievability (PoR)~\cite{shacham2008compact} as a proof-of-storage (for completeness, we discuss how these proofs work in Appendix~\ref{appendix:por}). For implementing this protocol, we use the bilinear BN256~\cite{naehrig2010new} group of elliptic curves for BLS signatures from the Kyber library~\cite{kyber}. As in Filecoin, by default, the payment for all valid PoRs is dispensed once the storage contract ends (we test different server payment modalities in the next section).

\textbf{Mainchain implementation.} The mainchain has one block type; mainchain blocks. Without \sysname, all transactions are logged in these blocks, but with \sysname, these blocks contain mainchain transactions along with concise summaries of the sidechain workload. Inspired by Filecoin, mining on the mainchain relies on the storage amount the servers provide, i.e., mining power. Thus, for each mainchain round, a miner will be selected to mine the next block in proportion to the amount of service it provides. This selection is done using cryptographic sortition similar to~\cite{Gilad17}. That is, each miner evaluates a Verifiable Random Function (VRF) using a random seed computed from the previous round and compares it to their mining power. If the VRF output is lower than the mining power, the miner is a potential round leader. It is possible to have multiple leaders in the same round. To resolve this issue, the miner with the first announcement will be the round leader.\footnote{This is a simplifying assumption that can be comparable to Algorand's priority-based work distribution~\cite{Gilad17}. Here, the priority is given to the first leader with a valid block.} We use Algorand's implementation of VRFs~\cite{github:algorand}. Recall that \sysname's performance is agnostic to the mainchain consensus, thus this simplified version is enough for our purposes.

\textbf{Traffic generation.} We apply the same traffic distribution observed in Filecoin using~\cite{site:filfox}. That is, 98\% of the traffic is service-related transactions and 2\% are payment (currency transfer) transactions. The majority of the service-related traffic is composed of proof-of-storage transactions (the rest are for contract-propose/commit transactions), which are proof-of-retrievability transactions in our case. Transactions are added at the beginning of each mainchain round. A detailed explanation of our traffic distribution and generation, as well as transaction structure and sizes, can be found in Appendix~\ref{apdx:traffic}.

\textbf{Optimistic rollups comparison.} To provide a baseline comparison with a layer-two solution, we implement an optimistic rollup solution inspired by optimism~\cite{optimism}, that we dub \opBoost and apply it to the base resource market. In each round, service-related transactions are batched up and processed off-chain. Then the state changes induced by a finalized batch are logged on the mainchain.

\section{Performance Evaluation}
\label{sec:perf-eval}
We now discuss our experiment setup and results showing the performance gains of \sysname.\medskip

\noindent\textbf{Experiment setup.}
We deploy our system on a computing cluster composed of 10 hypervisors, each running a 12-Core, 130 GiB RAM, VM, connected with 1 Gbps network link. Our experiments, unless stated otherwise, are conducted over a network of 8000 nodes, with each node serving two contracts. An experiment length is 61 mainchain rounds, with an epoch length $k = 10$ mainchain rounds---equivalent to 30 sidechain rounds (we use mc-round and sc-round to denote mainchain and sidechain round, respectively), with a committee size of 500 miners. As discussed earlier, 2\% of the traffic is payment transactions while the rest are service-related transactions, as in Filecoin. Our mainchain and sidechain block sizes are 1MB, and we collect experiment data using SQLite~\cite{website:sqlite}.

We test the impact of several parameters, including traffic distribution, block size, the ratio of sidechain rounds to mainchain rounds in an epoch, and various server payment modalities. For performance comparison purposes, we show results for the file storage network performance with and without \sysname (this will allow us to quantify the performance gains that our system can achieve).\footnote{To the best of our knowledge, our system is the first solution that target resource markets, and the first to introduce a dependent sidechain. Thus, it is incomparable to existing sidechain solutions that restrict themselves to two-way peg and independent sidechains.}

In reporting our results, we measure the following metrics (some experiments also report blockchain size):
\begin{itemize}
%\setlength\itemsep{-0.4em}
%\vspace{-4pt}
\item \emph{Throughput:} This metric corresponds to the number of transactions processed per a mainchain round. It is computed as the average number of transactions that appear in a mainchain block (or round) plus the average number of transactions in all sidechain blocks published during the same mainchain round.

\item \emph{Confirmation time:} This is the time (in mainchain rounds) that takes a transaction to be confirmed from the moment it enters the queue. As we use PBFT on both chains, once a transaction is published in a block it is considered confirmed.
\end{itemize}

We note that after an experiment run ends, we continue processing any traffic left in the queues until those queues are empty, and measure the metrics above accordingly.

\vspace{3pt}
\noindent\textbf{Scalability.} In this experiment, we investigate the behavior of \sysname as we increase the network load by increasing the number of service contracts. In each run, we set the contracts to be \{2K, 4K, 8K, 16K, 32K, 64K, 512K\}. Each contract results in a PoR transaction per mainchain round, a propose and commit contract transactions (when the contract is created), and a service payment transaction when the contract expires. We report throughput, latency, and blockchain size of the file storage market with and without \sysname. We also report the number of newly serviced (or active) contracts per round (note that a new contract becomes active once its contract-commit transaction is confirmed). The results are shown in Figure~\ref{fig:scale}.

\begin{figure}[tp]
    \centering
    \begin{subfigure}[t]{0.45\linewidth}
        \centering
        \includegraphics[width=\linewidth]{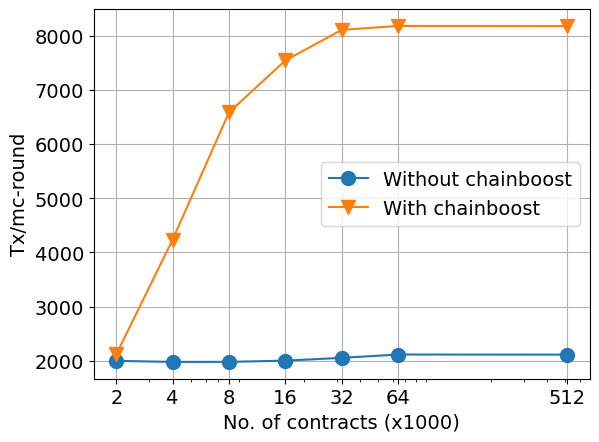}
        \caption{Throughput}
        \label{fig:throughput-scale}
    \end{subfigure}%
    \begin{subfigure}[t]{0.45\linewidth}
        \centering
        \includegraphics[width=\linewidth]{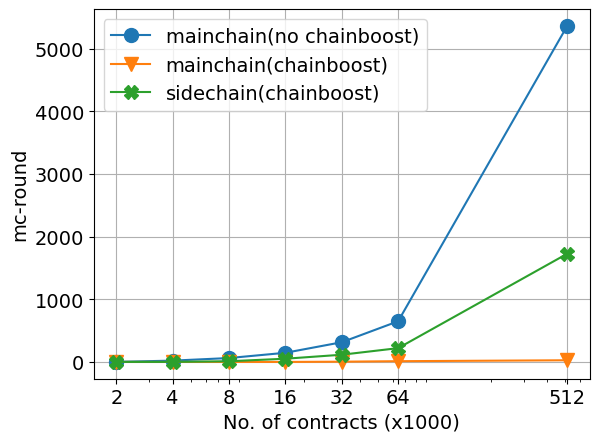}
        \caption{Confirmation time}
        \label{fig:latency-scale}
    \end{subfigure}
    \begin{subfigure}[t]{0.45\linewidth}
        \centering
        \includegraphics[width=\linewidth]{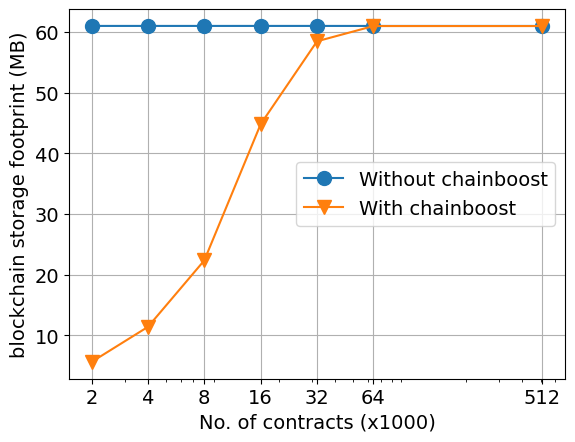}
        \caption{Storage footprint}
        \label{fig:footprint-scale}
    \end{subfigure}
	  \begin{subfigure}[t]{0.45\linewidth}
        \centering
        \includegraphics[width=\linewidth]{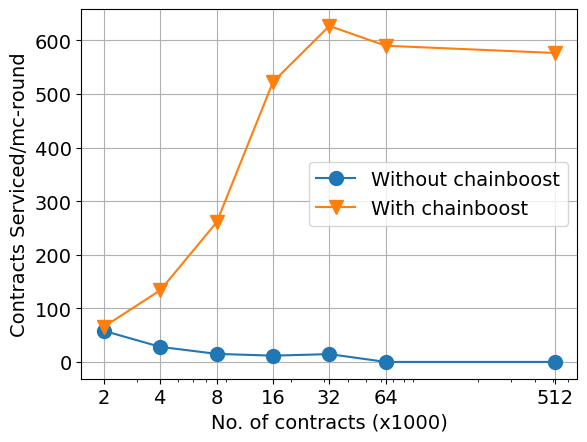}
        \caption{Serviced contracts}
        \label{fig:service-scale}
    \end{subfigure}
\caption{Scalability experiment results (for 1 MB block, and 3 sc-rounds per mc-round for \sysname).}
\label{fig:scale}
\vspace{-6pt}
\end{figure}

As shown, \sysname significantly improves the performance of the file storage. Throughput doubles as we double the number of contracts from 2K to 4K until it peaks at $\sim$4x the throughput of the base file storage market (i.e., the one without \sysname) for higher numbers of contracts. Improvements are even more noticeable when we examine the confirmation time as seen in Figure~\ref{fig:latency-scale}. Besides the almost-zero\footnote{In our implementation the mainchain and sidechain traffic is added at the beginning of the round. So an instant transaction publication, or almost zero confirmation delay, means a transaction that was queued and processed in the same round.} confirmation time on the mainchain we get by moving the frequent transactions to the sidechain, the confirmation time growth on the sidechain is not as steep as on the base file storage. Thus, even if the number of transactions increases significantly, they are processed significantly faster with \sysname.

For the blockchain size (Figure~\ref{fig:footprint-scale}), we observe that for networks with less than 32K contracts, \sysname reduces the amount of data stored on the blockchain, going from a $\sim$90\% reduction for 2K contracts to 0.3\% reduction for 32K contracts (for the latter, the huge load causes the network to be saturated with non-sidechain traffic, still with \sysname these transactions are confirmed faster than the base file system with \sysname). \sysname also allows more contracts to be serviced each round (Figure~\ref{fig:service-scale}), since \sysname allows the transactions that activate expired contracts to be confirmed at a fast rate. Overall, with \sysname, a file storage network can achieve higher throughput, lower confirmation time, and provide a faster service with a smaller blockchain storage footprint, thus greatly improving the market's scalability.\medskip

\noindent\textbf{Impact of block size.} We examine the impact of the mainchain and sidechain block sizes on performance; we set the size to $\{1, 2, 4, 8\}$ MB, and $\{0.5, 1, 1.5, 2\}$ MB for the mainchain and sidechain, respectively.

Throughput-wise, in Table~\ref{tbl:throughput-bs}, we observe that applying \sysname, even at the smallest block size setting, improves the throughput of the base system for small mainchain block sizes (1 and 2 MB). For a mainchain with 1MB block size, throughput is improved by a factor of 2.36x - 6.6x, which increases as the sidechain block size grows. For larger mainchain block sizes, we notice that \sysname with small block sizes has a smaller impact on throughput. This is because that traffic amount is fixed for all setups, the mainchain with a large block size can handle all this traffic. However, in practice, a mainchain block size of less than 8MB is usually used, and the results indicate that system designers are better off with using a large block size with the \sysname's sidechain.

From confirmation time, in Table~\ref{tbl:latency-bs}, we find that \sysname improves confirmation time significantly. This is because the queues are not clogged with PoR transactions so other transactions are confirmed almost instantly on the mainchain, and the latency on the sidechain is improved as its block size increases as it allows for processing a larger number of transactions per sc-round. 

\begin{table}[t]
    \begin{center}
        \caption{Average throughput (in $\times1000$ tx/mc-round) with and without \sysname for different block sizes.}
        \label{tbl:throughput-bs}
        \resizebox{\linewidth}{!}{
            \begin{NiceTabular}{ c c | c | c | c | c | c|}[vlines, corners=NW]
                \cline{3-7}
                & &  \Block{2-1}{w/o \sysname}   & \multicolumn{4}{c|}{ w/ \sysname---SC Block Size (MB)}\\
                \cline{4-7}
                & & & 0.5 & 1 & 1.5 & 2\\ 
                \midrule
                \Block{*-1}{\rotatebox[origin=c]{90}{\parbox{2cm}{\centering MC Block\\ Size (MB)}}}
                & 1 & 2.00 & 4.73 & 7.54 & 10.36 & 13.17\\  
                \cline{2-7}
                & 2 & 3.96 & 4.73 & 7.54 &  10.36 & 13.17 \\ 
                \cline{2-7}
                & 4 & 7.910 & 4.73 & 7.54 &  10.36& 13.17 \\ 
                \cline{2-7}
                & 8 & 15.98 &  4.73 & 7.54 &  10.36 & 13.17 \\ 
                \hline
            \end{NiceTabular}
        }
    \end{center}
    \vspace{-3pt}
\end{table}

\begin{table}[t]
    \begin{center}
        \caption{Average confirmation time (in mc-rounds) for mainchain/sidechain traffic under different block sizes.}
        \resizebox{\linewidth}{!}{
            \begin{NiceTabular}{ c c | c | c| c | c | c | c|}[vlines, corners=NW]
                \cline{3-8} 
                & & \Block{2-1}{ w/o \sysname}   & \multicolumn{5}{c|}{ w/ \sysname---SC Block Size (MB)}\\
                \cline{4-8}
                & &   & Chain& 0.5 & 1 & 1.5 & 2\\ 
                \midrule
                \Block{*-1}{\resizebox{6.8pt}{!}{\rotatebox[origin=c]{90}{MC  Block Size (MB)}}}
                & \multirow{2}{*}{1}  & \multirow{2}{*}{145.88} & MC & 0.03 & 0.03 & 0.03 & 0.03\\  
                \cline{4-8}
                & & & SC & 132.59 & 51.41 & 24.46 & 10.92\\  
                \cline{2-8}
                & \multirow{2}{*}{2} &  \multirow{2}{*}{61.81} & MC & 0 & 0 &  0 & 0 \\ 
                \cline{4-8}
                & & & SC & 132.59 & 51.41 & 24.46 & 10.92\\  
                \cline{2-8}
                & \multirow{2}{*}{4} & \multirow{2}{*}{20.74} & MC & 0 & 0 & 0 & 0\\
                \cline{4-8}
                & & & SC & 132.59 & 51.41 & 24.46 & 10.92\\  
                \cline{2-8}
                & \multirow{2}{*}{8} & \multirow{2}{*}{0.92} & MC & 0 & 0 & 0 & 0 \\
                \cline{4-8}
                & & & SC & 132.59 & 51.41 & 24.46 & 10.92\\  
                \hline
            \end{NiceTabular}
        }
        \label{tbl:latency-bs}
    \end{center}
    \vspace{-6pt}
\end{table}

When we apply the results of Table~\ref{tbl:latency-bs} to blockchains that have a similar block rate as Bitcoin (i.e. 1 block every 10 minutes) and a block size of 1MB, the confirmation time will go from around 1.02 days to an almost instant confirmation for mainchain transactions, to around 8, 4 and 1.8 hours for sidechain transactions, when using a 1, 1.5, 2 MB sidechain block sizes, respectively. If we apply the same analysis to chains using the Ethereum block rate (1 block every 12 seconds), the confirmation time on the sidechain drops from around 30 min to around 10, 5, and 2 minutes, for the same sidechain block sizes.

Overall, by using \sysname to process frequent and heavy service-related traffic, throughput and latency depend on the sidechain configuration. Thus, \sysname adopters should pick a value for the sidechain block size based on the expected traffic load.

\vspace{3pt}
\noindent\textbf{Number of sidechain rounds per epoch.}
In all previous experiments, an epoch lasted for 10 mainchain rounds, or 30 sidechain rounds. In this experiment, we use the same number of mainchain rounds per epoch but we change the number of sidechain rounds in an epoch to 40, 60, 80, and 100. (Note that a large number of sidechain rounds per epoch is possible with blockchains that have long rounds.) We measure performance for each configuration, repeating the experiment for different traffic loads.

As shown in Figure~\ref{fig:throughput-ratio}, throughput increases until the network is no longer saturated with transactions. For example, for a network with 8K contracts, throughput maxes out at around 8.4K transactions/mc-round, when the sidechain epoch size/mainchain epoch size (in rounds) ratio is 4. After that, we start observing empty blocks on the sidechain for higher ratios. We can extrapolate that for 32K, where throughput maxes out at a ratio of 16.

The sidechain confirmation time decreases as the ratio of sc-rounds to mc-rounds increases (Figure~\ref{fig:latency-ratio}). As the number of sidechain rounds per epoch increases, a larger number of transactions can be published, reducing the waiting time (in terms of mc-rounds) of service transactions in the sidechain queue. If we compare these numbers to the base market without \sysname, under the same configuration, from previous experiments, we observe that we can achieve around 4x - 11x improvement in throughput and 62.5\% - 94\% decrease in confirmation time. Thus, the ratio of sidechain rounds per mainchain round should be defined in a way that captures the expected load to achieve high-performance gains when \sysname is used.

\begin{figure}[t!]

    \centering
    \begin{subfigure}[t]{0.5\linewidth}
        \centering
        \includegraphics[width=\linewidth]{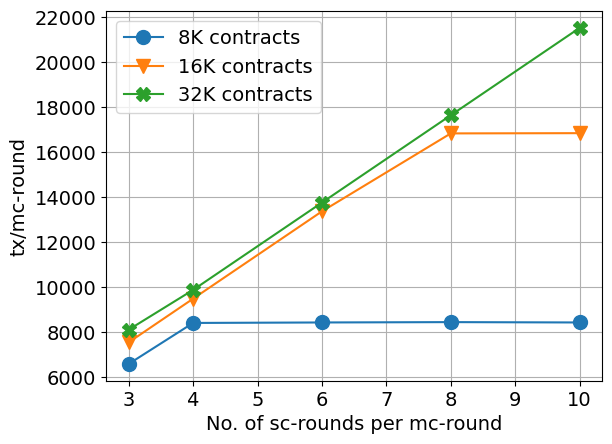}
        \caption{throughput}
        \label{fig:throughput-ratio}
    \end{subfigure}%
    \begin{subfigure}[t]{0.5\linewidth}
        \centering
        \includegraphics[width=\linewidth]{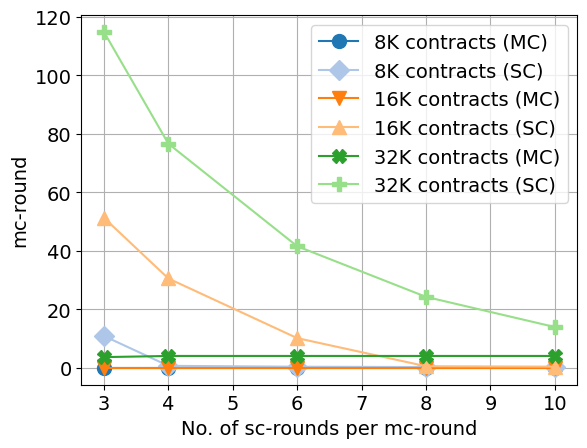}
        \caption{confirmation time}
        \label{fig:latency-ratio}
    \end{subfigure}
\vspace{-3pt}
\caption{\sysname under various sidechain-to-mainchain round ratios}
\vspace{-4pt}
\end{figure}

\begin{figure}[t!]

    \centering
    \begin{subfigure}[t]{0.5\linewidth}
        \centering
        \includegraphics[width=\linewidth]{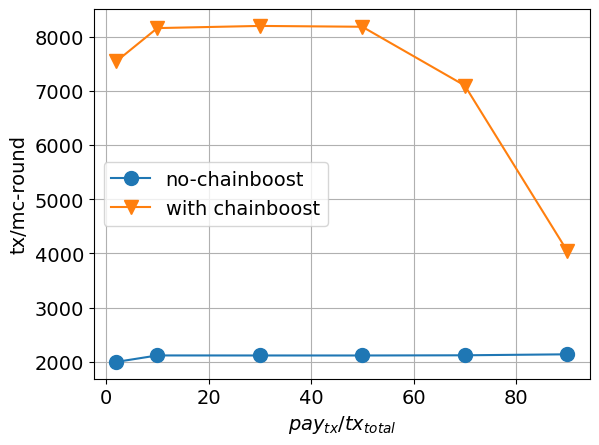}
        \caption{throughput}
        \label{fig:throughput-dist}
    \end{subfigure}%
    \begin{subfigure}[t]{0.5\linewidth}
        \centering
        \includegraphics[width=\linewidth]{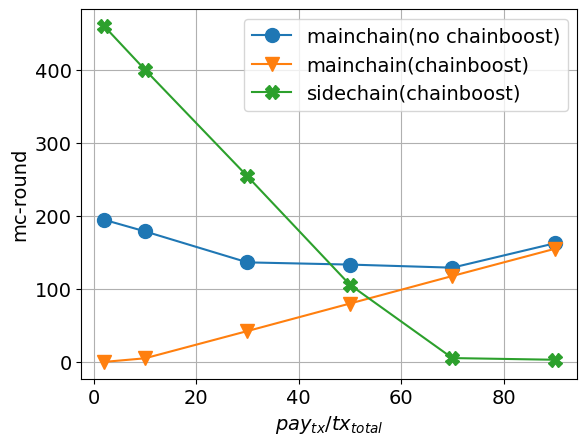}
        \caption{confirmation time}
        \label{fig:latency-dist}
    \end{subfigure}
\vspace{-3pt}
\caption{\sysname under various traffic distributions.}
\vspace{-6pt}
\end{figure}

\vspace{3pt}
\noindent\textbf{Traffic distribution.}
In this experiment, we fix the number of total transactions we generate but we vary the distribution of payment transactions, making them $\{2, 10, 30, 50, 70, 90\}\%$ of the total transactions while the rest are service-related ones. Throughput-wise, we observe in Figure~\ref{fig:throughput-dist} that \sysname achieves maximum throughput when payment transactions are within 10\% - 50\%, after which throughput drops. However, \sysname is still able to achieve 2x throughput improvement when payment transactions reach 90\% of the total traffic. This is due to the difference in size between payment transactions and service transactions; space occupied on the mainchain by service transactions without \sysname can fit more payment transactions, thus increasing throughput.

For confirmation time, we observe in Figure~\ref{fig:latency-dist} that as the ratio of payment transactions increases, confirmation time on the mainchain increases, while the one on the sidechain decreases. This is due to the increase in the mainchain workload while the sidechain one decreases. Hence, for blockchain systems where non-service-related traffic is dominant, \sysname provides marginal improvements. Nonetheless, in practice, distributed resource market have a large workload of service transactions, which is where \sysname should be used.

\vspace{3pt}
\noindent\textbf{Server payment modalities.}
Paying for each proof-of-service delivery results in a huge number of service payment transactions, which overwhelm the system. Thus, markets aggregate these payments before processing (either pay at the end of a contract, as in Filecoin, or resort to deploying additional machinery). In this experiment, we show how \sysname can support payment aggregation at low overhead. In particular, we test three service payment modalities: pay at the end of the contract duration, at the end of each mainchain round, and pay at the end of each epoch (note the last two are equivalent in the \sysname as summaries are produced at the end of an epoch). So in the case without \sysname, and although there are no epochs, we enforce them by paying every $k$ mainchain rounds (lower frequency than every round and higher than at the end of a contract as servers would like to receive their partial payments earlier than at the end of, especially, long contracts). We run this experiment with 8000 miners (or servers) with each miner handling one contract.

Our results (Table~\ref{tbl:modality}) show that \sysname improves the performance of the base file storage market. It achieves better throughput, with similar confirmation time if we pay for service at the end of each epoch, rather than the end of the contract. Thus, even if the base system aggregates payments as part of its network protocol, \sysname can further improve this performance and allow for flexible setup as this frequency can be changed over time by merely changing the summary rules rather than the market network protocol. Hence, \sysname supports (flexible) micropayment aggregation by design.

\begin{table}[t!]
\caption{Throughput and confirmation time for various server payment models with/without \sysname.}
\resizebox{\linewidth}{!}{
\begin{tabular}{ | c |c | c | c | c | c |}
\cline{2-5}
  \multicolumn{1}{c|}{} & Payment Modality Setup  &  Throughput & MC Latency & SC Latency\\ 
 \hline
 \multirow{3}{*}{\rotatebox[origin=c]{90}{without}}& contract end & 1978.70 & 22.23 & N/A\\  
 \cline{2-5}
 & epoch end  &  2001.63 &  22.56 & N/A \\ 
 \cline {2-5}
 &  each mc-round  & 2202.20 &24.65 & N/A \\
 \hline
 \hline
 \multirow{2}{*}{\rotatebox[origin=c]{90}{with}} & contract end & 6583.20 & 0 & 8.39 \\  
  \cline {2-5}
 & epoch end  &  6964.49 & 0.84 & 8.23  \\ 
 \hline
\end{tabular}
}
\label{tbl:modality}

%\vspace{-4pt}
\end{table}

\vspace{3pt}
\noindent\textbf{Optimistic rollups comparison.} We compare \sysname to a resource market that uses an Optimism-inspired rollup solution (\opBoost). We configure the rollup to process 1.5 MB of transactions (an Optimism batch is 1.8 MB~\cite{optimism-size}, we use 1.5 MB for a fair comparison with \sysname), and to take 3 mc-rounds for processing (Optimism batch processing takes between 2 and 4 Ethereum rounds to be finalized~\cite{jumpcryptoBridgingFinality}; we use an average of 3 Ethereum rounds). \sysname is configured with 3 sc-rounds per mc-round, and each sidechain block is 0.5 MB. We report throughput, latency---the time needed for a transaction to appear in a sidechain block in \sysname, or a non-finalized processed rollup in \opBoost, and finality---the time needed for a rollup or sidechain summary to be considered final.

\begin{table}[t!]
    \caption{Comparison of \sysname vs. \opBoost}.
    \centering
    \small
    \begin{tabularx}{\columnwidth}{|>{\centering}m{.185\linewidth}|
    *{4}{Y|}}
    \cline{2-4}
        \multicolumn{1}{c|}{}  &  Throughput \footnotesize{(tx/mc-round)}  & Latency (mc-round) & Finality (mc-round)  \\
    \hline
        \sysname & 4730 & 132.59 & 132.59 \\
    \hline
        \opBoost  &  4682 & 129.81 & 50529 \\
    \hline
    \end{tabularx}
    \label{tab:op}
\vspace{-6pt}
\end{table}

Our results (Table~\ref{tab:op}) show that, given the equivalent configurations, throughput and latency are similar. However, \sysname provides significantly shorter finality than \opBoost since a transaction is final once it appears in a meta-block. \opBoost has a one-week contestation period, which is equivalent to 50400 mc-round in the experiment setup; users has to wait for that long before taking any action based on the rollup-induced state changes. Furthermore, \opBoost has the risk of adopting invalid results due to incentive incompatibility of verifiers who may not contest the rollup results as discussed earlier.

\section{Related Work}
\label{sec:related-work}
\vspace{-4pt}
In this section, we review prior work on sidechains and relevant blockchain scalability solutions.\medskip

\textbf{Sidechain-based solutions.} As mentioned before, the term sidechain was first introduced by~\cite{Back14}; they outline the main use cases of sidechains and present a (high-level) protocol for two-way peg. Since then, a long line of work around sidechains has emerged.

Among these works,~\cite{Back14,Gavzi19,Kiayias19,Connor17,Garoffolo18,Garoffolo20} focus on two-way peg protocols, where~\cite{Kiayias19} allow the sidechain to monitor events on the mainchain (a restricted form of data exchange).~\cite{Croman16} briefly note the challenges of employing a sidechain, and~\cite{Robinson2022atomic,Robinson20merits} tackle permissioned sidechains for Ethereum. \cite{Lee19} utilize sidechains to build a mutable blockchain, where a sidechain has the solo purpose of mining a modified block to replace an older one. The works~\cite{lee2021hierarchical,Rovzman21} use sidechains to improve scalability:~\cite{lee2021hierarchical} combine sharding and sidechains to allow parallel processing of smart contracts, and~\cite{Rovzman21} use a tree of sidechains to scale a distributed manufacturing service. Besides requiring smart contracts in both schemes, \cite{lee2021hierarchical} rely on the users to distribute the load among the shards, and \cite{Rovzman21} improve scalability by compromising security (by operating the sidechain at lower mining difficulty than the mainchain). Other works, mainly~\cite{Yin2021sidechains, Gavzi19}, use a committee selected at random from a pool of sidechain miners to advance the sidechain and synchronize its state in the mainchain, but they both are limited to asset exchanging protocols and use independent chains.~\cite{de2022hierarchical} also employ independent sidechains in synchronizing the asset transfer.

In parallel, many industrial initiatives explored sidechains~\cite{Poon17,cosmos,Wood16,btcrelay,liquid,rootstock,xdai}. However, these suffer from several limitations: they work only with specific blockchain types, mainly Bitcoin and Ethereum; focus only on two-way peg; rely on optimistic-rollups; target permissioned blockchains or federated ones; or have incomplete specifications, and the majority requires smart contracts. Also, all previous solutions, except for~\cite{Gavzi19, Kiayias19, Yin2021sidechains}, do not provide rigorous security analysis. Even those who included formal security notions, they were limited to two-way pegs and independent sidechains. 

None of these (independent) sidechain solutions can be trivially extended to permit arbitrary data exchange, workload sharing, and chain pruning. Independent chains have their own ecosystems in terms of miners/users, network protocol, transactions, and tokens, which makes it hard (if not unfeasible) to split the workload between them while keeping the mainchain as the single source of system state; clearly having only a two-way peg operation does not support that. At the same time, chain dependency means that their valid and secure operations are tied to each other; thus, additional countermeasures are needed to ensure that both chains satisfy safety and liveness, and that they can recover from interruptions on any of these chains.

\textbf{Relevant blockchain scalability solutions.} Existing solutions can be classified into three categories \cite{Hafid2020scaling}: layer-one or on-chain solutions, layer-two or off-chain solutions, and sidechain solutions (reviewed earlier). On-chain solutions improve performance by changing the blockchain parameters (e.g. block size and block time) or modifying its design (e.g. changing the consensus mechanism or blockchain structure) such as sharding~\cite{Kokoris18}. Changing the blockchain parameters is problematic. For example, increasing the block size may increase throughput, but it increases verification and propagation delays (and still all transactions are logged on-chain), leading to higher network partitioning probability---a solution that is not favorable by the community and leads to hard forks (e.g., the case of BSV~\cite{bsv}). Sharding divides the miners and the workload among shards to allow parallel processing, and they usually focus on how to handle cross-shard transactions and load balancing among the shards. None of these schemes devised a way to prune the mainchain

Off-chain solutions include two classes: state channels and rollups. In both cases transactions are processed off-chain and only state changes are recorded on-chain. Payment channel networks~\cite{Decker15,Green17,Miller19} are an example of the state channel approach, but they target only payment aggregation. ZK rollups~\cite{starkware,Garoffolo20,Bowe20,bonneau2020coda} extend this model by performing computations off-chain with ZK proofs submitted on-chain to prove output validity. ZK proofs are costly and may require a trusted setup. Optimistic rollups~\cite{Poon17,matic} assume all submitted results are valid unless proven otherwise. A set of computing parties perform the computation, one of them submit the results (or state changes), while the rest of the parties (aka verifiers) should dispute incorrect results using an interactive challenge-response protocol. Despite its efficiency, this paradigm requires long dispute periods, slowing down system operation and service delivery. Furthermore, it may involve centralized trusted verifiers~\cite{op-verifier, arbitrumStateArbitrums}, while incentivizing non-trusted verifiers to perform their job is still an open question~\cite{li2023security, mediumCheaterChecking, mediumOptimisticRollup} which may lead to adopting incorrect ledger state changes.

\sysname aims to exploit the full potential of sidechains. It considers dependent sidechains and allows for arbitrary data exchange with the mainchain, rather than just two-way peg. It achieves validity of state changes resulting from the sidechain using fast consensus. Thus, it voids the use of expensive ZK proofs or slow and interactive contestation process. \sysname offers a unified interface and architecture that can be used with any resource market regardless of the service it provides. Also, it is agnostic for the consensus and mining type used by the underlying resource market mainchain. Lastly, \sysname is the first system to show how sidechains can be used for blockchain pruning, thanks to its two-type block architecture.

\section{Conclusion}
\label{conclusion}
\vspace{-4pt}
In this paper, we presented \sysname, a secure performance booster for blockchain-based resource markets. It introduces a new sidechain architecture that has a mutual-dependence relation with the mainchain. \sysname allows these chains to exchange arbitrary data---all heavy/frequent service-related transactions are offloaded to the sidechain, while logging only concise summaries on the mainchain. Furthermore, \sysname supports blockchain pruning and automatic recovery from interruptions. We analyze the security of our system and conduct thorough experiments to evaluate its performance. The results show the great potential of \sysname in improving performance of resource markets in practice.

\section*{Acknowledgements}
This material is based upon work supported by the National Science Foundation under Grant No. CNS-2226932. 

\bibliographystyle{plain}
\bibliography{chainBib}

\appendices
\section{Autorecovery Protocol: A Detailed Description}
\label{appendix:autorecovery}
In this section, we elaborate on the design of \sysname's autorecovery protocol. We define all sidechain interruption cases (using the ABC threat modeling framework~\cite{Almashaqbeh19}), then we present our protocol that addresses them.

\subsection{Interruption Cases}
We begin with modeling the tolerance of PBFT-based consensus. Then, we list the possible threats categorized by their underlying causes. Finally, we remove threats that are neutralized by system assumptions and elaborate on those that require detection and recovery.

\textbf{PBFT consensus tolerance.}
In PBFT~\cite{castro1999practical}, up to $f$ faulty nodes can be tolerated (i.e., while preserving the safety and liveness properties) with a committee size, denoted as $\cs$, of $3f+1$. Thus, reaching an agreement requires at least $2f+1$ supporting votes (to guarantee safety---agreeing on valid records---these votes are from honest nodes). On the other hand, if $f+1$ votes are absent, liveness will be violated as no agreement can be reached. The original PBFT~\cite{castro1999practical} is not scalable for large consensus groups, and so follow-up works build their consensus protocols on modified versions. ByzCoin~\cite{Kogias16}, for example, is designed to scale better and provides a near-optimal tolerance with $\cs = 3f+2$; so it requires at least $2f+2$ supporting votes to reach agreement. We adopt this version since we deal with large-scale distributed systems. Also, we adopt the leader-based approach, so a leader proposes a block for the committee to agree on, and once an agreement is reached, the block is appended to the sidechain.

\textbf{Potential threat cases.}
Interruptions on the sidechain are related to absence of meta or summary blocks (due to not initiating the agreement or not reaching an agreement), absence of sync-transactions, or publishing invalid blocks/sync-transactions. These are threats to the liveness and safety of the sidechain, and consequently the mainchain since service-related traffic is not being processed properly. These could be due to malicious or lazy leader or committee members---those who may go unresponsive, behave arbitrarily, or simply passively collaborate with the adversary by accepting records without validation.\footnote{Note that being unresponsive due to bad network conditions is ruled out in our system since we assume bounded-delay message delivery.}

We denote the committee elected in an epoch to run the sidechain as primary committee $\com_0$ of size $\cs$. We also denote the number of malicious members in a committee as $m$ and the number of lazy nodes as $l$. Accordingly, we have the following threat cases:
\begin{itemize}
    \item Lazy/malicious committee: if the majority of the committee members are lazy or malicious, the following threats may take place: becoming an unresponsive committee, publishing invalid blocks or invalid sync-transactions, or pretending that the leader is unresponsive to force a leader change.
    \item Lazy/malicious leader: who proposes invalid (meta or summary) blocks/sync-transactions or is unresponsive (i.e., does not initiate agreement on block proposals or does not issue a sync-transaction).
\end{itemize}

The unresponsive lazy committee threat happens when lazy and malicious members are unresponsive, so that when $m+l \geq f+1$. This prevents committee agreement, and hence, violates liveness. The threat of publishing invalid blocks happens when lazy/malicious members agree on an invalid block proposal. That is, the leader is malicious and proposes an invalid (meta or summary) block; malicious (perhaps colluding) committee members would agree, and lazy honest members (who do not validate the block) would also agree. This threat happens when $m+l \geq 2f+2$, thus violating safety. Lastly, a malicious leader may attempt to issue an invalid sync-transaction---one that is not consistent with its corresponding summary-block. Note that for this threat case, the invalid sync-transaction will be rejected by mainchain miners since only valid transactions are accepted. Thus, this threat is detected by the absence of $\tx_{\sync}$ from the mainchain ledger state $\led_{\mainc}$ (as shown in Figure~\ref{fig:autorecovery-pro}).

For changing the leader, as we describe later, \sysname uses the view change mechanism of PBFT~\cite{castro1999practical}. That is, committee members agree that the leader is unresponsive/malicious and force a leader change. This agreement also requires $2f+2$ votes to take place. Thus, framing a leader may happen when $m+l \geq 2f+2$, which is considered a safety violation (i.e., merged with the safety case above). At the same time, not reaching an agreement on changing the leader puts the system on hold. This takes place when there are unresponsive $m + l \geq f+1$ committee members, which is a liveness violation (covered by the liveness case above).

\subsection{Protocol Design}
We introduce an autorecovery protocol to address the previous threat cases, guaranteeing that the sidechain can recover from any interruptions, thus preserving the security of the mainchain and the market as a whole. At a high level, our protocol is composed of two techniques to handle malicious/unresponsive leader or committee. The former is detected by the primary committee members themselves, and dealt with by initiating a \emph{leader change}. While for the latter, we introduce the concept of \emph{backup committees} who will monitor the primary committee and take over if a misbehavior is detected.

\textbf{(1) Leader change.}
This mechanism is composed of two phases: detection of a malicious/unresponsive leader, and recovery from this situation. Leader change is parameterized by a timeout $\zeta$ and it is repeated for every round in every epoch. 

\emph{Detection.} In the primary committee, each member monitors the leader's behavior and can independently detect the following:
\begin{itemize}
    \item If no activity is detected for a $\zeta$ duration, then the leader is identified as unresponsive.

    \item If an invalid block proposal is received, or an invalid sync-transaction is broadcast, the leader is identified as malicious.
\end{itemize}

For the latter, recall that all committee members process the sidechain traffic, and hence, can verify the validity of any proposed meta-block. Also, all of them have copies of the meta-blocks published so far in an epoch, so they can validate a proposed summary-block. Similarly, they can validate the sync-transaction since they have its corresponding summary-block (same applied for mainchain miners). Recall also that an invalid sync-transaction will be rejected by mainchain miners.

\emph{Recovery.} Upon the detection of either of the above cases, each (honest) committee member sends a \textit{leader-change} message to the rest of the committee. Once the next candidate leader receives $2f+1$ leader-change messages (including its own), it takes over as the new leader in that epoch. The leader-change messages, which in case of invalid behavior will be accompanied by a copy of the invalid block proposal/sync-transaction, constitute the \textit{leader-change certificate} that is published in the next meta-block (and also the summary-block of the epoch) as a proof on the validity of the leadership change.

Note that, if $f + 1 \leq m+l \leq 2f+1$, an agreement on changing a leader will not be reached. This is reduced to the threat of unresponsive-committee (handled by the backup committee technique discussed next). Furthermore, it could be the case that a large number of future leader candidates are malicious or lazy. This will result in a long process of leader change that will put the system on hold (violates liveness). This situation is also reduced to the unresponsive committee case. Furthermore, absence of a valid sync-transaction from the mainchain, due to continuous leader-change, will be reduced to the unresponsive committee case and addressed based on whether there is a valid summary-block in the epoch.

\textbf{(2) Backup committees.}
In this technique, besides the primary committee handling the workload of an epoch, there will be a set of backup committees standing by. A backup committee has the same size as the primary committee and monitors the sidechain traffic. The primary committee and all backup committees (we denote their number by $\kappa$) are composed of disjoint set of miners. Thus, for each epoch the election mechanism produces $\kappa + 1$ committees $\{\com_i\}_{i = 0}^{\kappa+1}$ with a publicly known order---$\com_0$ is the primary committee, $\com_1$ is the first backup committee, and so on. Note that when a backup committee takes over, it becomes the primary committee, and the next committee in line becomes the first backup committee, and so on. This mechanism is also composed of two phases: detection and recovery, as we explain below.

\emph{Detection.} A backup committee is responsible for detecting an unresponsive primary committee or one that publishes invalid blocks. Detecting the former is based on a timeout $\eta$ such that $\eta > \zeta$; that is, no blocks or sync-transactions are published during $\eta$ period. Detecting the latter is done when verifying a newly published block. The backup committee monitors the sidechain traffic and published blocks, so it can identify if a block is invalid.

\emph{Recovery.} For the case of lack of progress, the leader of the backup committee initiates an agreement on a \emph{unresponsive-committee} message. While if invalid blocks are detected, it issues a \emph{misbehaving-committee} message containing a reference to the invalid block(s). Once the agreement is complete, by having at least $2f+2$ signatures or votes over any of these messages, the backup committee will take over. All relevant information about the committee takeover will be published in the next meta-block (which will be created by this new committee), and eventually it will be published in the summary-block of that epoch. All invalid blocks are dropped: if it is a summary-block (and so all meta-blocks are valid), the new primary committee only publishes a summary-block~\footnote{In this case, there may be a race condition between the the transaction published by the faulty committee and the one from the backup committee. This can be remediated by delaying the processing of sync transactions on the mainchain by $\Delta$ (the delay bound on messages, as specified in Section~\ref{sec:prelim}), to check for the eventuality of race conditions. The sync transaction from the committee who took over will attest to the takeover---i.e., it will report the incident and show the agreement on its occurrence, and hence, will automatically invalidate the other sync-transaction}, otherwise, the work done by the faulty committee is considered invalid and the backup committee starts all over. In case of the absence of a sync-transaction at the end of an epoch, but there is a valid summary-block, no takeover will take place. Instead, the new epoch committee will issue a mass-sync-transaction covering both epochs.

Since a backup committee itself could be unresponsive or malicious, we require having $\kappa$ backup committees. Thus, the protocol above is repeated but with a larger timeout. That is, backup committee $\com_{i}$ will monitor for $i \cdot \eta$ period (as shown in Figure~\ref{fig:autorecovery-pro}) to account for the timeout needed for the primary committee and all backup committees ahead of $\com_i$.

The success of detection and recovery is conditioned on having the backup committee reach agreement on the misbehavior, and having this committee work properly when taking over. Thus, the success probability of the backup committee mechanism determines the success probability of our autorecovery protocol (recall that the unsuccessful leader-change threat is reduced to the unresponsive committee threat). In Appendix~\ref{appendix:autorecovery-sec-analysis}, we analyze the success probability of the autorecovery protocol, which allows us to set up the number of backup committees $\kappa$ needed to make this probability overwhelming, thus preserving liveness and safety of the sidechain, and consequently the mainchain.

\section{Limitations, Tradeoffs, and Future Work}
\label{appx:limit}
\textbf{Applicability.} In designing \sysname, we target a particular system category; blockchain-based resource markets. The level of performance gains depends on the underlying system design and functionality. \sysname is more suitable for systems that exhibit high volumes of frequent service-related traffic that can be summarized, and thus can be off-loaded to the sidechain. Systems that do not exhibit this behavior will not benefit much from the use of our scheme (as we show empirically in Section~\ref{sec:perf-eval}). Thus, outside the blockchain-based resource markets, the main motivation for this work, a viability study is needed in order to decide whether \sysname is applicable/beneficial. As part of our future work, we aim to identify other categories of blockchain-based systems where we can employ \sysname. In particular, we will investigate decentralized finance (DeFi) systems, e.g., decentralized exchanges and money lending applications, and tokens on top of Ethereum.

Furthermore, our system model targets permissionless blockchains. This begs the question of whether \sysname can be used in the context of permissioned blockchains. For example, a hybrid digital service has its own centrally-managed infrastructure, e.g., Akamai~\cite{akamai} for content delivery network (CDNs), and may rely on a peer-assisted network to extend its coverage   
and help in handling peak demands. Here, it makes more sense to have a permissioned blockchain to manage the peer-assisted network by creating a resource market for the service provided by end users (or peers). \sysname can be applied here, and in fact, at a lower overhead than what is incurred in a permissionless setting. That is, our autorecovery protocol is not needed; the miners managing the main and side chains are composed of a small set of reliable and fully-identified machines. This also means that committee election is not needed where the PBFT consensus is run by this set of registered miners. 

\textbf{The autorecovery protocol.} Based on the analysis in Appendix~\ref{appendix:autorecovery-sec-analysis}, there is a tradeoff between the committee size and the number of backup committees required to have a negligible failure probability of autorecovery. Naturally, larger committee size will lead to higher probability of having (active) honest majority, which leads to a smaller $\kappa$ (i.e., smaller number of backup committees, and thus shorter worst case autorecovery time). However, a larger committee size will introduce additional overhead in terms of increased agreement time. Thus, system designers should select a tardeoff based on the round length in the targeted resource market and any performance constraints it has with respect to service delivery.

\textbf{Miner incentives.} Incentive compatibility is an important parameter when studying blockchain security. An issue that arises in our setup is incentivizing miners to work faithfully on both the main and side chains. This is outside the scope of this work; we focus on the system design and its robustness/security aspects. As part of our future work, we plan to investigate the following direction: tying the incentives of maintaining both chains together. In detail, a miner (who is a member of the sidechain committee) will collect the full rewards of the recent blocks it has mined on the mainchain only if this miner is active on the sidechain. If it does not participate in the voting process on meta- or summary-blocks, or does not propose a block when selected as a leader, the mining rewards will be deducted based on the number and weight of missing votes (i.e., a vote on a summary-block may have a larger weight than one on a meta-block, also a leader's punishment will be larger than a member's punishment).

\textbf{Storage pricing and transaction fees.} Another issue that arises under the setup of \sysname is configuring the transaction fees. Note that not all transactions occupy permanent space on the mainchain, and even not all of them occupy a permanent space on the sidechain but instead their summaries do. Thus, sidechain transactions---these that are discarded---may have lower fees than mainchain ones, or these that are recorded in full in a summary-block, as they do not occupy a permanent storage space. At the same time, the fees of sidechain transactions can be used to solve the issue of miner incentives above; that is, sidechain traffic may have larger fees than these of mainchain traffic to compensate the miners for the work on the sidechain. Studying this tradeoff for resource pricing and transaction fees falls under the recent notion of multidimensional fees~\cite{diamandis2023designing}, and we leave it as part of our future work.

\section{Proof of Theorem 1}
\label{appx:chainboost-security}
In order to prove Theorem 1, we prove the following lemmas showing that \sysname is secure and so it preserves safety and liveness of the underlying resource market.

\begin{lemma}[Preserving safety]
\sysname preserves the safety property of the underlying resource market.  
\end{lemma}

\begin{proof}
    Due to meta-block pruning and the mainchain syncing process that \sysname introduces, the following threats may arise in the resource market:
    \begin{itemize}
        \item Service slacking and theft: A server, who did not provide a service, claims it did so but the proofs were removed during sidechain pruning. Also, a customer, who received the service, claims it did not receive anything as no proofs exists (since they were pruned) and avoids paying the server.
    
        \item Out-of-sync mainchain: A committee leader does not issue a $\tx_{\sync}$ at the end of the epoch, causing the sidechain and mainchain to be out of sync.
    
        \item Invalid syncing: A committee leader issues an invalid $\tx_{\sync}$, thus causing the mainchain to be updated with invalid state changes. 
    
        \item Violating sidechain quality: A sidechain committee mines invalid meta- or summary-blocks, resulting in invalid mainchain state updates during syncing. 
    \end{itemize}

    Note that \sysname does not impact currency transfers or any other transactions that are processed on the mainchain. These are secure by the security of the underlying resource market and its mainchain. 
    
    We now show how \sysname addresses these threats. This mainly relies on the use of a secure PBFT with a secure/unbiased committee election mechanism, and the security of the autorecovery protocol of our scheme. Furthermore, in Appendix~\ref{appendix:autorecovery-sec-analysis}, we derive a lower bound for the committee size $\cs$ to guarantee honest majority.

    \emph{Service slacking and theft}.  These threats are mitigated by \sysname's design. First, recall that the sidechain processes the service-related traffic using the same specifications/logic of the underlying resource market. And second, meta-blocks are not dropped until after their summary-block is mined and its corresponding $\tx_{\sync}$ is confirmed on the mainchain, and that summary-blocks are stored permanently on the sidechain. Therefore, sidechain miners can verify that a $\tx_{\sync}$ is consistent with its corresponding summary-block. Given that the sidechain uses a secure PBFT consensus, the committee size guarantees honest majority, and the autorecovery protocol addresses any interruptions on the sidechain, published summary-blocks are valid based on their corresponding meta-blocks. Thus, only valid summaries, that a committee with honest majority has produced through consensus, will be used to sync the mainchain. This is enough to verify any service claims despite old meta-block pruning. 

    \emph{Out-of-sync mainchain}. These threats are mitigated using \sysname's autorecovery protocol and mass-syncing. When the sidechain committee detects that no $\tx_{\sync}$ is issued, they initiate a leader-change and the new leader issues the $\tx_{\sync}$. Also, if this primary committee does not do that, it is considered unresponsive and the backup committee technique will be activated. Mass-syncing handles cases when there are rollbacks on the mainchain, allowing the mainchain to capture all summarized traffic that was missed during the rollback period.  

    \emph{Invalid $\tx_{\sync}$ and violating sidechain quality}. Invalid $\tx_{\sync}$ is detected immediately when mainchain miners verify its consistency with the corresponding summary-block, and so it will be rejected. The autorecovery protocol will deal with the leader who issues such an invalid transaction, as well as a leader that proposes an invalid summary- or meta-block, or a committee that agrees on such invalid proposals as explained in Section~\ref{sec:robustness} and Appendix~\ref{appendix:autorecovery}.

    As a result, \sysname satisfies safety for the sidechain, and hence, does not result in any security threats to the safety of the mainchain and the resource market as a whole.
\end{proof}

\begin{lemma}[Preserving liveness]
\sysname preserves the liveness property of the underlying resource market.  
\end{lemma}

\begin{proof}
    Having a committee managing the sidechain and responsible for processing and summarizing service-related traffic may impose threats to the liveness and public verifiability of the resource market as follows:
    \begin{itemize}
        \item DoS attacks: The sidechain committee targets particular clients or servers and excludes their transactions from being published.
        
        \item Violating sidechain liveness: A misbehaving sidechain committee does not add new blocks to the sidechain, or misbehaving leader does not issue a $\tx_{\sync}$. This will disrupt the entire resource market as service-related transactions are not (fully) processed. 
        
        \item Violating sidechain/mainchain public verifiability: This threat encompasses any impacts of the summary/syncing/pruning processes on the public verifiability of the side and the main chains. 
    \end{itemize}

    Again, since \sysname uses a secure PBFT consensus to run the sidechain with a rotating committee, and a secure autorecovery protocol, it addresses the threats above as follows.
    
    \emph{DoS:} This is alleviated by the security of the PBFT consensus and its rotating committee. As the sidechain committee changes every epoch, maintaining a position to perpetually deny particular transactions is unfeasible.

    \emph{Violating sidechain liveness}: These threats are mitigated using \sysname's autorecovery protocol. At the round level, when the committee does not hear anything from the leader, they issue a leader change. At the epoch level, when a backup committee does not see growth on the sidechain after a timeout, it engages in a committee-change and takes over as described in Section~\ref{sec:robustness} and Appendix~\ref{appendix:autorecovery}.

    \emph{Violating sidechain/mainchain public verifiability}: \sysname uses a PBFT-based consensus for the sidechain, and relies on autorecovery to ensure the correct operation of this sidechain. Also, any rollbacks on the mainchain will be handled using the mass-syncing mechanism, so no summaries will be lost. As summary-blocks are accurate snapshots of the state of the sidechain for their epochs, by the security of the PBFT consensus, and that meta-blocks are kept until the $\tx_{\sync}$ is confirmed (and even for the rollback period), public verifiability of both the sidechain and mainchain state is preserved.

    Accordingly, \sysname satisfies liveness of its sidechain, and hence, preserves liveness of the mainchain and the resource market.
\end{proof}

\noindent\textbf{Security of the autorecovery protocol.} 
For the threats solved through autorecovery, the mitigation is not unconditional; autorecovery fails if all backup committees selected for an epoch fail, and a backup committee fails if it has over $\lthresh$ deviant members (the liveness threshold $\lthresh$ is the threshold for absent votes that will violate liveness as defined in Appendix~\ref{appendix:autorecovery-sec-analysis}).
Through a probabilistic analysis of the autorecovery failure probability, we derive a formula for failure probability as a function of the committee size $\cs$ and the number of backup committees $\kappa$. An appropriate configuration keeps the failure probability negligible.
\begin{lemma}\label{theorem-sec}
The probability of the event that \sysname's autorecovery fails (denoted as $AF$) can be expressed as:
\begin{align*}
\Pr[AF] = \sum_{i=(\kappa+1) \lthresh}^{(\kappa+1) \cs} {\frac{\binom{\mathscr{M}}{i}\binom{N-\mathscr{M}}{(\kappa+1) \cs - i}}{\binom{N}{(\kappa+1) \cs}}}
    {\left(\frac{[y^{i}]\Psi(y)}{\binom{(\kappa+1) \cs}{i}}\right)}
\end{align*}    
\end{lemma}
\noindent where $\mathscr{M}$ is the number of misbehaving miners (so $\mathscr{M} = m+l$ from the previous section), $N$ is the number of all miners, $\Psi(y) = \left( \sum_{i=\lthresh}^{\cs}~\binom{\cs}{i} y^i \right)^{\kappa+1}$, and $[y^{i}]\Psi(y)$ denotes the coefficient of $y^{i}$ in $\Psi$, calculated using $[y^{i}]\Psi(y) = \frac{1}{i!} \odv*[order=i]{\Psi(0)}{y}$. 

We provide a detailed security analysis and formal proof of Lemma~\ref{theorem-sec} in Appendix~\ref{appendix:autorecovery-sec-analysis}.

\section{Autorecovery Security Analysis}
\label{appendix:autorecovery-sec-analysis}
\sysname's autorecovery protocol must ensure that the market can recover from any of the sidechain interruptions (defined in Appendix~\ref{appendix:autorecovery}) with overwhelming probability. In other words, the autorecovery failure probability must be negligible. In this section, we formally model the components related to the autorecovery protocol, define its failure probability, and derive a formula for this probability that allows system designers to configure the protocol parameters in a way that makes this probability negligible to achieve secure operation. This section also includes a formal analysis of the committee size to guarantee honest majority (i.e., meeting the safety threshold for PBFT).

\subsection{Modeling and Definitions}
\label{sec:defs}

\noindent\textbf{Modeling PBFT consensus tolerance.} 
Recall that \sysname uses a PBFT to run the sidechain as in~\cite{Kogias16}. Also, as discussed in Section~\ref{sec:robustness}, \sysname uses the original PBFT's leader change mechanism~\cite{castro1999practical} to replace a misbehaving or unresponsive leader. We define the liveness and safety tolerance (or thresholds) for both of these PBFT protocols along the lines as in~\cite{castro1999practical,hafid2022tractable}.

The original PBFT provides optimal safety and liveness thresholds, which we define as follows (where $\cs$ denotes the committee size).

\begin{definition}[Optimal Safety Threshold]\label{optimal-safety-threshold}
The optimal safety threshold, denoted by $\sthresh^\prime$, is the minimum number of supporting votes required to agree in a PBFT consensus, where $\sthresh'=\lceil\frac{2\cs+1}{3}\rceil$.
\end{definition}

Note that the optimal safety threshold indicates the minimum number of nodes that can approve a leader's proposal, whether it is a malicious agreement or an honest one. That is, if the number of misbehaving nodes is larger than $\sthresh'$, then a malicious agreement (e.g., on an invalid proposal that does not comply with the designated protocol) can be made. To achieve security, and so have an honest agreement, the number of honest responsive nodes must be at least $\sthresh'$. 

\begin{definition}[Optimal Liveness Threshold]\label{optimal-liveness-threshold}
The optimal liveness threshold, denoted by $\lthresh'$, is the minimum number of absent votes that will prevent the committee from agreeing in a PBFT consensus, where $\lthresh'=\cs-\sthresh'+1= \lfloor\frac{\cs+2}{3}\rfloor$.
\end{definition}

For the sidechain PBFT, and without loss of generality, we assume that system designers adopt a PBFT-based consensus protocol that requires at least $\sthresh$ ($\sthresh \in [\sthresh^\prime, \cs)$) supporting votes to consider a leader's proposal approved. We define the general safety and liveness thresholds for this PBFT-based consensus protocol as follows. Note that the thresholds related to the adopted PBFT-based consensus are bounded by the optimal thresholds but can be different with respect to the chosen consensus.

\begin{definition}[Safety Threshold]\label{safety-threshold}
The safety threshold, denoted by $\sthresh$, is the minimum number of supporting votes required to agree in a PBFT-based consensus, where $\sthresh \in [\sthresh^\prime, \cs)$.
\end{definition}

\begin{definition}[Liveness Threshold]\label{liveness-threshold}
The liveness threshold, denoted by $\lthresh$, is the minimum number of absent votes that will prevent agreement in a PBFT-based consensus, where $\lthresh=\cs-\sthresh+1$.
\end{definition}

Note that a higher safety-threshold means larger part of the committee should vote to have an agreement, so an optimal PBFT needs less threshold than other PBFT-based consensus mechanisms. On the liveness side, the higher a consensus is resilient against liveness attacks, the more members need to be non-responding to violate liveness, so optimal for liveness means a higher threshold. Thus, we have the following ordering: 
\begin{equation}\label{thresholds}
    \lthresh \leq \lthresh' \leq \sthresh' \leq \sthresh
\end{equation}
This ordering is used in the rest of this section to merge or reduce threat cases and to extract the threat scenarios and committee election outcomes that violate security.\\

\noindent\textbf{Modeling committee election outcome.}
We consider a network of size $N$, each node (aka miner) can be either honest, lazy, or malicious with probability $p_h$, $p_l$, $p_m$, respectively, with $p_h+p_l+p_m=1$. By assumption, each (primary and backup) committee member is elected independently. We have one primary committee and $\kappa$ backup committees, all denoted by $\com_i$, where $\com_0$ is the primary committee, and $\com_i$ for $i \in \{1, \dots, \kappa\}$ are the backup committees, for an epoch. Let $h_i$, $l_i$, and $m_i$ be the number of honest, lazy, and malicious members in a committee $\com_i$ of size $\cs$ for $i\in \{0, \dots, \kappa\}$, such that $h_i + l_i + m_i  = \cs$.

%\ghada{is it small x or capital X?}
Let $\mathcal{O}_i=(h_i, l_i, m_i)$ denote the random variable corresponding to committee election outcome for $\com_i$. The total number of malicious or lazy nodes, referred to as misbehaving going forward, in $\com_i$ is denoted by $x_i$. The overall committee election outcome, denoted by $\mathcal{O}$, is a set of committee election outcomes for $\{\com\}_{i=0}^{\kappa}$:
\begin{align}
    \mathcal{O}=\bigl\{(h_i, l_i, m_i)\bigr\}_{i=0}^{\kappa}=\{\mathcal{O}_i\}_{i=0}^{\kappa}
\end{align}
Let $\Pr[\mathcal{O}]$ denote the probability of having $\mathcal{O}$ as the overall committee election outcome and the probability that the committee election outcome for $\com_i$ is such that it contains $X_i$ malicious or lazy members is shown by
$\Pr[{O}_i | X_i=m_i+l_i]$. \\

\noindent\textbf{Committee and autorecovery failure definitions.} A committee $\com_i$ fails when it cannot reach an agreement or agree on malicious proposals. In turn, autorecovery fails when all (primary and backup) committees fail. That is, (1) when the primary committee is unresponsive (does not reach an agreement) or agrees on malicious proposals, happens when Equation~\ref{pl} holds, (2) all backup committees fail to detect the misbehaving primary committee, happens when Equation~\ref{bcf} holds, or (3) after it takes over the primary committee, this backup committee becomes a misbehaving one itself, happens when Equation~\ref{bcl} holds.\footnote{Note that any malicious strategic behavior is covered in this combination, including malicious nodes which constitute the backup committee's majority but do not prevent it from the initial detection of the threat, but after taking over the primary committee behave maliciously---by confirming invalid state updates or not responding}
\begin{align}\label{pl}
\mathcal{O} =\bigl\{\mathcal{O}_0| \lthresh \leq x_0 < \sthresh \lor x_0 \geq \sthresh\bigr\} 
\end{align}

\vspace{-10pt}
\begin{align}\label{bcf}
\mathcal{O} =\bigl\{\{\mathcal{O}_i\}_{i=1}^\kappa |x_i \geq \lthresh'\bigr\} 
\end{align}

\vspace{-10pt}
\begin{align}\label{bcl}
\mathcal{O} =\bigl\{\{\mathcal{O}_i\}_{i=1}^\kappa | \lthresh \leq x_i < \sthresh \lor x_i \geq \sthresh\bigr\}
\end{align}

\noindent Therefore, based on the threshold ordering in Equation~\ref{thresholds}, the auto-recovery fails when the overall committee election outcome equals equation~\ref{backup-committee-mechanism-failure-probability} (so the lower bound for the number of misbehaving nodes $x_i$ is $\lthresh$ as defined below).

\begin{align}\label{backup-committee-mechanism-failure-probability}
    \mathcal{O}_F = 
    \bigl\{\{\mathcal{O}\}_{i=0}^{\kappa} |x_i \geq \lthresh\bigr\}
\end{align} 

Based on that, we define a committee failure and the autorecovery failure (AF) as follows.
\begin{definition}[Committee Failure]\label{committee-failure}
A (primary or backup) committee $\com_i$ fails if in its committee election outcome, $\mathcal{O}_i=(h_i,l_i,m_i)$ the number of misbehaving nodes violates the liveness threshold, i.e., $x_i \geq \lthresh$.
\end{definition}

\begin{definition}[Autorecovery Failure (AF)]\label{autorecovery-failure}
\sysname's autorecovery protocol fails if all primary and backup committees, $\com_i$ for $i \in \{0, \dots, \kappa\}$, fail (cf. Definition~\ref{committee-failure}), and the probability of autorecovery failure, $\Pr(AF)$, is defined as $\Pr(AF) = \Pr(\mathcal{O}_F)$ where $\mathcal{O}_F$ is as defined in Equation~\ref{backup-committee-mechanism-failure-probability}.
\end{definition}

Accordingly, a secure autorecovery protocol is one that has a negligible $\Pr(AF)$. In the next section, we derive an expression for $\Pr(AF)$ that allows system designers to configure the protocol parameters in a way that makes this quantity negligible, thus achieving security.

\subsection{Autorecovery Failure (AF) Probability}
\label{sec:af-analysis}
As defined in the previous section, autorecovery fails when all (primary and backup) committees fail. Since these committees are disjoint, the probability of failure in one committee can be generalized to other committees. That is, the population and probabilities change with each draw, and so we model our committee election as a sampling without replacement. We follow two approaches in our analysis of the failure probability: the joint Hypergoemtric distribution approach, and the generating function approach, which we describe below (the former is easier to understand, while the latter is easier to compute).\medskip

\noindent\textbf{Joint Hypergeometric distribution approach.} Let $N$ be the total number of miners in the network, and $\mathcal{M}$ be the total number of misbehaving miners among $N$. We model the distribution of misbehaving nodes in each committee as a hypergeometric distribution, as each draw decreases the population and changes the probability of selecting a misbehaving node to the committee.

\begin{theorem} 
\sysname's autorecovery failure probability can be expressed as:
\begin{multline*}
\label{th1}
\Pr(AF) = \\
\sum_{x_\kappa=\lthresh}^{\cs}\dots\sum_{x_0=\lthresh}^{\cs}\prod_{i=0}^{\kappa}  \frac{\binom{\mal - \sum\limits_{j=0}^{j<i} {x_j}} {x_i} \binom{(\all-i\cs) - (\mal - \sum\limits_{j=0}^{j<i} {x_j})}{\cs - x_i} } {\binom{\all-i\cs}{\cs}}
\end{multline*}
\end{theorem} 
\begin{proof}
The number of malicious or lazy nodes in the first committee is modeled by the hypergeometric distribution with the parameters $N$, $\mathscr{M}$ and $\cs$. The probability of having an election outcome in which the primary committee has $x_0$ misbehaving nodes is
\[ \Pr(\mathcal{O}_0|X_0=x_0)=\frac{\binom{\mathscr{M}}{x_0}\binom{N-\mathscr{M}}{\cs-x_0}}{\binom{N}{\cs}} \]

For the backup committees, we need to account for the nodes that were selected in the committees preceding them as well as for changes in the misbehaving nodes population $\mathscr{M}$; in the $i^{th}$ committee, the overall population is $N - i\cs$ and the malicious nodes are $\mathscr{M} - \sum_{j=0}^{j=i-1} x_j$. Thus, the probability of the $i^{th}$ committee containing $x_i$ misbehaving members would be:
\[ \Pr(\mathcal{O}_i|X_i=x_i)= \frac{\binom{\mal - \sum_{j=0}^{j= i-1} {x_j}} {x_i} \binom{(\all-i\cs) - (\mal - \sum_{j=0}^{j = i-1} {x_j})}{\cs - x_i} } {\binom{\all-i\cs}{\cs}} \]

The overall election outcome which causes autorecovery failure is the one in which all committees fail. Recall that a committee fails when the total number of malicious or lazy nodes in that committee becomes equal to or greater than the committee liveness-threshold, $\lthresh$ (see Definition~\ref{committee-failure}), so we have $\lthresh \leq x_i \leq \cs$. Thus:
\[ \Pr(AF)= \sum_{x_\kappa=\lthresh}^{\cs}\dots\sum_{x_0=\lthresh}^{\cs}\prod_{i=0}^{\kappa}\Pr(\mathcal{O}_i | X_i=x_i) \]

By replacing $\Pr(\mathscr{O}_i | X_i=x_i)$ with its expression computed using the Hypergeometric distribution from above, we obtain the formula in the theorem, which completes the proof. 
\end{proof}

\noindent\textbf{Generating function approach.} The previous formula is complicated and hard to evaluate (needs extensive computations). Hence, we do the same analysis using the generating function approach. In~\cite{hafid2022tractable}, the generating function approach is used to calculate the probability of a sharding take-over attack in a sharding-based protocol in which the attack is successful once at least one shard becomes insecure.

The goal of this analysis is to calculate the number of ways of having all committees fail in order to compute the failure probability of autorecovery. 

First, we calculate the probability of having $x$ malicious or lazy elected committee members, in all $\kappa+1$ committees (so $x = x_0 + \dots + x_{\kappa}$). Then we calculate the ways we can distribute them among $\kappa + 1$ committees such that all of them fail. Finally, we compute the auto-recovery failures by multiplying the two previous probabilities. 

The generating function that represents a failure in one committee is the following:
\begin{align}\label{eq1}
    \Phi(y)=\sum_{x_i=\lthresh}^{\cs}{\binom{\cs}{x_i}\times y^{x_i}}
\end{align}
\noindent where each coefficient $\binom{\cs}{x_i}$ indicates the number of ways $x_i$ malicious or lazy nodes can be distributed in that committee. Note that Equation~\ref{eq1} represents a committee's failure because its number of malicious or lazy nodes is more than the liveness-threshold threshold $x_i \geq \lthresh$.

Since the committees do not share members, they can be represented as a union of disjoint sets. The convolution rule~\cite{lehman2010mathematics} can be used to calculate the generating function of all committees $\Psi(y)$ as the product of their generating functions, i.e., $\Psi(y)=(\Phi(y))^{\kappa+1}$. The coefficient of $y^x$ for $(\kappa+1)\times \lthresh \leq x \leq (\kappa+1)\cs$ represents all the ways through which $x$ misbehaving nodes can be distributed among $\kappa + 1$ committees such that all those committees fail. Hence, we have:
\begin{equation}
	[y^{x}]\Psi(y) = \frac{1}{x!} \odv*[order=x]{\Psi(0)}{y}
\end{equation}

\begin{theorem}
\sysname's autorecovery failure probability can be expressed as follows:
\[\Pr(AF) =
    \sum_{x=(\kappa+1)\lthresh}^{(\kappa+1)\cs}{\frac{\binom{\mathscr{M}}{x}\binom{N-\mathscr{M}}{(\kappa+1)\cs - x}}{\binom{N}{(\kappa+1) \cs}}}
    \cdot
    {\frac{[y^{x}]\Psi(x)}{\binom{(\kappa+1)\cs}{x}}}\]
\end{theorem}
\begin{proof}

The probability of all committees' failure, knowing that $(\kappa+1) \lthresh \leq x$ misbehaving are elected from the network, is the following.\footnote{Note that the probability of all committees' failure, knowing that ${x}<(\kappa+1)\lthresh$ malicious or lazy members are elected from the network, is zero.}
\[
    \alpha = \Pr(\cs,\lthresh,\kappa,x) = \frac{[y^{x}]\Psi(y)}{\binom{(\kappa+1)\cs}{x}}
\]
Also, the probability of having a total of $x$  malicious or lazy members elected from the network for the total of $\kappa+1$ committees, each with the size of $\cs$ is
\[
    \beta = \Pr(\cs,\kappa,x) = 
    \frac{\binom{\mathscr{M}}{x}\binom{N-\mathscr{M}}{(\kappa+1)\times \cs - x}}{\binom{N}{(\kappa+1) \cs}}
\]

Finally, the autorecovery failure probability, considering the probability of having a total of $x\in\{(\kappa+1)\lthresh, \dots, (\kappa+1)\cs\}$  misbehaving miners elected from the network, and the probability of having all committees failed beacause of these miners, is the following:
\begin{equation*}\label{eq-final}
\begin{aligned}
	    \Pr(AF) &= \Pr(\cs,\lthresh,\kappa) = 
    \sum_{x=(\kappa+1) \lthresh}^{(\kappa+1)\cs}\alpha\beta
\end{aligned}
\end{equation*}

Substituting for $\alpha$ and $\beta$ defined earlier produces the formula in the theorem, which completes the proof.
\end{proof}

\subsection{Recovery Time}
\label{subapdx:Recovery-Time}

In this section, we analyze the time needed for the autorecovery protocol to detect and recover from worst-case interruptions. This happens when the primary committee produces all the meta-blocks during an epoch, with the last meta-block being invalid. Each meta-block takes $\eta - \epsilon$ (for some infinitesimally small $\epsilon > 0$), either because the leader is withholding it for that duration, or the committees are engaged in leader change operations. The following $\kappa - 1$ backup committees behave similarly since they start over when an invalid block is detected.

The total delay imposed by the adversarial behavior on the primary committee or any of the first $\kappa-1$ backup committees would be:
\begin{equation}
   \mathcal{D}_c= (\eta-\epsilon) (k-1) + T_{agr}
\end{equation}

\noindent where $k$ is the number of sidechain blocks in an epoch, and $T_{agr}$ is the time needed for the backup committee to agree on an \textit{misbehaving-committee} message.

In the worst-case scenario, the total delay in the first $\kappa$ committees would be:
\begin{equation}
\{\mathcal{D}\}_0^{\kappa-1}= \kappa \mathcal{D}_c
\end{equation}

The final backup committee will perform the entire work (mining metablocks and a summary-block), each block taking $\eta - \epsilon$ to produce. Thus its delay would be:

\begin{equation}
\mathcal{D}_\kappa=  (\eta-\epsilon) k
\end{equation}
As a result, in the worst case, the time it takes for the autorecovery to recover from is (expressed as a function of the parameters defined above):
\begin{equation}
\label{eqn:delay}
\begin{aligned}
		    \mathcal{D}(k, \eta,\kappa)&= \{\mathcal{D}\}_0^{\kappa-1} + \mathcal{D}_\kappa \\
    &=(\kappa+1) (\eta-\epsilon) (k - 1) \\&+\kappa (T_{agr} ) + (\eta-\epsilon)
\end{aligned}
\end{equation}

\noindent $\eta$, being the timeout for a committee change operation, is bounded by the sidechain round duration $r_{sc}$. Also, $T_{agr}$ is significantly smaller than $r_{sc}$. Thus Equation~\ref{eqn:delay} becomes:
\begin{equation*}
\begin{aligned}
\mathcal{D}(k, \eta,\kappa)&= (\kappa+1) (\eta-\epsilon) (k - 1) +\kappa (T_{agr} ) \\& \:\:\:\: + (\eta-\epsilon) 
                                    \\&\leq (\kappa+1) r_{sc} (k-1) + \kappa r_{sc} + r_{sc}
                                    \\&\leq r_{sc} \bigl((\kappa+1) (k-1) + \kappa + 1\bigr)
                                    \\&\leq r_{sc} (\kappa+1) k = (\kappa+1) \Delta_{epoch}
\end{aligned}
\end{equation*}
where $\Delta_{epoch}$ is the epoch duration. Considering that the adopted committee election outcome is resilient against adaptive adversaries in the $\tau$-agile corruption mode (or slowly adaptive adversary in which $\tau$ steps are needed to corrupt parties)~\cite{pass2017hybrid,avarikioti2019divide}, this means that $\tau \approx\kappa+1$ epoch duration to prevent the adversary from targeting the committee members after being elected.

\subsection{Committee Size Analysis}
\label{one-committee-security}
In this section, we analyze the committee size $\cs$ with the goal of limiting the probability of having a misbehaving majority in the primary committee. In particular, we formulate the probability of this event and bound it to an acceptable failure $F$, allowing us to determine $\cs$ that satisfies $F$ (recall that the primary committee and all backup committees have the same size $\cs$).

We draw the members of the committee from a finite population of $N$ miners, such that $N$ contains $\mathscr{M}$ misbehaving miners and a miner cannot be elected twice. Thus, the committee election can be modeled as a sampling without replacement from the population $N$. The probability of a primary committee containing $x_0$ misbehaving nodes would follow the Hypergeometric distribution:
\begin{equation}
    \begin{aligned}
\Pr(X=x) =  \frac{\binom{\mal} {x} \binom{\all - \mal}{\cs - x} } {\binom{\all}{\cs}}
    \end{aligned}
\end{equation}
The primary committee fails when $x > \theta_l$ (as defined in Definition~\ref{committee-failure}). The probability of this event can be computed as:
\begin{equation}
    \begin{aligned}
    	\Pr(X \geq \theta_l) = \sum_{x=\theta_l}^{\cs} \frac{\binom{\mal} {x} \binom{\all - \mal }{\cs - x } } {\binom{\all}{\cs}}
    \end{aligned}
\end{equation}

In a network with a large node population $N$, as assumed in~\cite{algorand} and~\cite{Kogias16}, committee election can be modeled as a sampling with replacement, thus can be modeled using a Binomial distribution. As such, the probability of a failed committee can be approximated as:
\begin{equation}
    \Pr(X \geq \theta_l) = \sum_{x=\theta_l}^{\cs} \binom{\cs}{x}p^x (1-p)^{\cs - x}
\end{equation}
 
Both Binomial and Hypergeometric distributions can be bound using a Chernoff bound~\cite{hariharan2012chernoff}. Let $p$ be the probability of a miner to be a misbehaving one, so $p = p_m + p_l$ from before, we define $\mu = p\cs $ to be the expected number of misbehaving nodes and so we can express $\theta_l = \omega + \mu$ for $\omega > 0$ (where we use $\omega$ instead of $\Delta$ from~\cite{hariharan2012chernoff} to avoid overloading) leading to $\Pr(X \geq \theta_l) = \Pr(X \geq \omega + \mu)$, which can be bounded as:
\begin{equation}
    \begin{aligned}
        \Pr(X \geq \omega + \mu) \leq 
            \begin{cases}
                \exp(\frac {-\omega^2}{3\mu}) & if \frac{\omega}{\mu} < 1 \\
                \exp(\frac{-\omega}{3}) & otherwise
            \end {cases}
    \end{aligned}
    \label{eqn:chernoff-1}
\end{equation}

We define $\gamma = \frac{\lthresh}{\cs}$, so we have $\omega + \mu = \gamma \cs$, and hence Equation~\ref{eqn:chernoff-1} becomes:
\begin{equation}
    \begin{aligned}
       \Pr(X \geq \omega + \mu) &= \Pr(X \geq \gamma \cs) \\ &\leq 
            \begin{cases}
                \exp(\frac {-(\gamma - p)^2 \cs}{3p}) & if \frac{\gamma}{p} < 2 \\
                \exp(\frac{-(\gamma-p)\cs}{3}) & otherwise
            \end {cases}
    \end{aligned}
    \label{eqn:chernoff-2}
\end{equation}

For $\gamma = \frac{1}{3}$, and an adversarial power between 25\% and 30\%, $\frac{\gamma}{p}$ is less than 2 ($\frac{4}{3}$ and $\frac{10}{9})$ respectively. We use $\exp(\frac {-(\gamma - p)^2 \cs}{3p})$ to bound the failure probability that we set to be $F=10^{-10}$. $\exp(\frac{-(\gamma - p)^2 \cs}{3p}) \leq 10^{-10}$ means that  $(\gamma-p)^2 \cs \geq 30 p \ln(10)$ and $\cs \geq \frac{30 p \ln(10)}{(\gamma-p)^2}$. Thus, the committee size, in this case, needs to be at least between $2487$ and $18651$. For failure probabilities $F_1 = 10^{-5}$ and $F_2 = 10^{-3}$, the committee size needs to be at least between $1244$ and $9326$, and at least between $747$ and  $5595$, respectively.

\section{Compact Proofs of Retrievability} 
\label{appendix:por}
The compact proof of retrievability (PoR) introduced in~\cite{shacham2008compact} is based on BLS signatures~\cite{boneh2001short}. It is a challenge-response protocol; the verifier challenges the prover to compute a proof based on a randomly selected parts (or blocks) of the stored file. Without these blocks, the prover cannot generate a valid proof. This scheme is proven to be secure against any PPT adversary under the Computational Diffie–Hellman (CDH) assumption over bilinear groups in the random oracle model.

\textbf{Notation.} We denote the natural numbers by $\NN$, the integers by $\ZZ$, and the integers modulo some prime $p$ by $\ZZ_p$. A cyclic multiplicative group $\GG$ has an order $p$ generated by generator $g\in \GG$. We let $\lambda \in \NN$ denote the security parameter. A bilinear group is given by a description $(e, p,\GG,\GG_T)$ such that $\GG$ and $\GG_T$ are cyclic groups of order $p$ and $e: \GG \times \GG \to \GG_T$ is a bilinear asymmetric map (pairing), which means that $\forall a, b \in \ZZ_p : e(g^a, g^b) = e(g, g)^{ab}$. Thus, we implicitly have that $e(g, g)$ generates $\GG_T$. The membership in $\GG$ and $\GG_T$ can be efficiently decided, group operations and the pairing $e(\cdot, \cdot)$ are efficiently computable, generators can be sampled efficiently, and the descriptions of the groups
and group elements each have linear size. Lastly, we have the hash function $ H:\{0, 1\}^*\rightarrow{\GG}$, which is modeled as a random oracle.

\textbf{Preprocessing phase.} In this scheme, and before the protocol can start, the client does the following: It breaks the erasure encoded file $M$ into $n$ blocks, each containing $s$ sectors. The file sectors are denoted by $\{m_{ij}\}$ for $1 \leq i \leq n$ and $1 \leq j \leq s$ (each sector $m_{ij} \in \ZZ_p$). So if the processed file is $b$ bits long, then it will be divided into $n = \lceil b/(s \log p)\rceil$ blocks. Second, the client generates a signature keypair $(\sk, \pk)$, chooses a random $\alpha \in \mathbb{Z}_p$, and computes $v = g^\alpha$. The client sets the secret key as $(\alpha, \sk)$ and the public key as $(v, \pk)$.

After that, the client computes the \textit{authenticator values} $\sigma_i$ for each block $i = \{1, \dots, n\}$ as:
    \begin{equation}
        \sigma_i = \Bigg(H\left(name \parallel i\right).\prod_{j=1}^{s}u_j^{m_{ij}}\Bigg)^\alpha 
    \end{equation}
where $name \in \mathbb{Z}_p$ is a random file name, and public generators $\{u_j\}_{1 \leq j \leq s}$ are randomly sampled from $\GG$ for each file. The client then computes a \textit{file tag} $\tau = name \parallel n \parallel u_1 \parallel \dots \parallel u_s$ together with a signature on $\tau$ using her private key $\sk$. Finally, the client sends the file sectors $\{m_{ij}\}$ and the authenticators ${\sigma_i}$ to a storage server.

\textbf{PoR protocol operation.} It proceeds as follows:
\begin{enumerate}
    \item Generate challenge: A verifier verifies the signature on the file tag $\tau$ using the client's public key $\pk$, and if valid, parses it to recover $name$ and $\{u_j\}_{1 \leq j \leq s}$. Then, it chooses a random challenge set $Q$ of index–coefficient pairs $\left\{\left(i,\ v_i\right)\right\}$ (where $i \in \{1, \dots, n\}$ and $v_i \in \ZZ_p$), and sends this set to the prover.
    
    \item Generate response: The prover calculates the response $\left(\sigma, \mu \right)$ where $\mu=\{\mu_j\}_{1 \leq j \leq s}$, and $\mu_j = \sum_{(i,\ v_i)\in Q}{v_im_{ij}}$. That is, it combines the blocks in $Q$ sector-wise, each with its multiplier $v_i$. The prover then computes the aggregated authenticator for the blocks in the challenge $\sigma = \prod_{(i,v_i)\in Q}{\sigma_i^{v_i}\in G}$, and sends $(\mu, \sigma)$ to the verifier.
    
    \item Verify response: The verifier computes the pairing and accepts the proof if $e(\sigma,g) = y$:
    \begin{equation}
        y = e\Bigg(\prod_{(i,\ v_i)\in Q}{H(name\parallel i)}^{v_i}.\prod_{j=1}^{s}{u_j^{\mu_j},\ v}\Bigg)
    \end{equation}
\end{enumerate} 

\vspace{3pt}
\textbf{Non-interactive PoR protocol in our proof-of-concept implementation.} The protocol described above is interactive, which does not suit the blockcahin setup. We can transform into an non-interactive one as follows: the prover computes the random challenge set using a random seed (for a security parameter of $\lambda = 128$ the seed should be 256 bits). The source of randomness for the seed can be pulled from the hash of latest block in the blockchain (since we model hash functions as random oracles). So, instead of having the verifier transmit the challenge set, for each round, the prover computes this set using the hash of the block mined in the previous round. Thus, in our setup, the client is the file owner who preprocesses the file as mentioned above, the blockchain miners are the verifiers, and the storage servers are the provers. The client put the value $\tau$ in the contract-propose transaction.

\textbf{Parameter selection.} As described above, one authenticator value for each file block is a storage overhead for the storage server. Using the bilinear BN256 group of elliptic curves in our implementation, authenticator values are points on this curve and $|\sigma_i| = \log p$. Therefore, having $s$ as the number of sectors, the total storage overhead is $\frac{b}{s \log p} \cdot \log p=\lceil\frac{b}{s}\rceil$. Choosing a larger $s$ value reduces the one-time overhead of servers' storage at the cost of higher communication. This communication is the PoR transactions that the server repeatedly sends to the blockchain miners in order to prove that it still storing the file. Each PoR proof includes one point and $s$ scalars, so the proof size is $(1+s) \log p$. Thus, a PoR transaction contains the PoR proof, along with the round number for which it was generated. In our implementation, we set $s=2$, thus the size of the PoR transaction is 515 bytes.

For example, for a $10$ GB file, with $s = 1$, the server storage overhead is $10$ GB (for the authenticator, so total storage is 20 GB for both), but the server response (PoR) contains one scalar in $\ZZ_p$ ($256$ bits) plus a point (64 bytes) on the curve $\GG$. Whereas, when $s = 10$, with the same file size, the server storage overhead would be $1$ GB for the authenticator (so total storage is 11 GB) but the server response contains $10$ scalars in $\ZZ_p$ ($2560$ bits) and a point on the curve $\GG$. Note that the PoR transaction size is a function of $s$, not the file size. Also, the query set size does not impact PoR transaction size or the storage overhead. This set is generated locally by the server and anyone can verify it using the round seed, the contract file-tag $\tau$, and the client public key, which are all publicly available.

\section{Transaction Types and Traffic Generation}
\label{apdx:traffic}

In this section, we provide details on the structure of the transactions and the way traffic is generated and processed by the miners in our proof-of-concept implementation.

\textbf{Transaction types.}
We have six transaction types as follows (their sizes are found in Table~\ref{tab:size-table}):
\begin{description}
\item[$\mathsf{Contract\mbox{-}propose}$:] Corresponds to a service agreement proposed by a client. Its structure contains the contract ID, a pointer to a payment structure (the escrow), a duration, a file size, a start round for the service, service fee (which is the payment value per mainchain round), and a file tag for the file to be stored.

\item[$\mathsf{Contract\mbox{-}commit}$:] A server issues a commit transaction when it commits to providing a service to a client; a response to a $\mathsf{Contract\mbox{-}propose}$, so it contains the contract ID.

\item[$\mathsf{Payment}$:] The payment (currency transfer) transactions in our system follow the UTXO model and have the same structure as in Bitcoin~\cite{site:bitcoin-tx}.

\item[$\mathsf{Proof\mbox{-}of\mbox{-}retrievability}$:] It is issued by a server every round for each file it stores. It contains the round number during which a PoR is generated, and a cryptographic proof based on the scheme from~\cite{shacham2008compact}.

\item[$\mathsf{Storage\mbox{-}payment}$:] This transaction pays the server from the client's escrow. It has the same structure as payment transaction, but adds the contract ID.

\item[$\mathsf{Sync\mbox{-}transaction}$:] Its structure follows the summary rules adopted in \sysname. In our proof-of-concept implementation, we summarize PoRs by counting them. Thus, this transaction contains an array of contract IDs (where we assume each contract is handled by one server) and the number of PoRs generated for that contract (by the corresponding server) for the epoch. Thus, its size varies based on the number of active contracts in the epoch.
\end{description}

\textbf{Traffic generation.}
When configuring our experiments, we assign each server a contract(s) with a specific duration of service. While the contract is active, the server issues one PoR transaction per mainchain round. The contract service duration (in mainchain rounds) is selected at random from a normal distribution $\mathcal{N}(\mu = 40,\sigma^{2} = 20)$ with a mean $\mu$ and standard deviation $\sigma$. Once a contract expires, mainchain miners issue a \textsf{Storage-Payment} transaction to pay the server for the service. Since we automatically match clients and servers, \textsf{contract-propose} and \textsf{contract-commit} are issued at the same time. This means that the contract is activated once again.

For traffic distribution, we follow the one of Filecoin~\cite{filecoin}. Using the the Filfox Explorer~\cite{site:filfox}, we find that around 2\% of all transactions in the system are currency transfer (called $\mathsf{Send}$ in the explorer). The remainder is split into 96\% representing service proof transactions (e.g., in the explorer these include $\mathsf{TerminateSectors, ProveCommitSector}$, etc.) and 2\% for other operational transactions (e.g., $\mathsf{ChangeOwnerAddress}$). So, in terms of our system, these translate into 2\% payment transactions, and 98\% PoR and propose/commit contract transactions.

\begin{table}[t!]
\centering
\caption{Transaction Types and Sizes}
\label{tab:size-table}
\begin{tabular}{|p{0.32\columnwidth}|p{0.58\columnwidth}|}
\hline
Transaction Type & Size (Bytes)\\
\hline \hline
\textsf{contract-propose} &  645 \\ \hline
\textsf{contract-commit}  &  79  \\ \hline
\textsf{Payment} &   398\\ \hline
\textsf{Proof-of-Retrievability} & 515 \\ \hline
\textsf{Storage-Payment} & 406 \\ \hline
\textsf{Sync-transaction} & Varies based on number of active contracts\\
\hline
\end{tabular}
\vspace{-6pt}
\end{table}

\textbf{Traffic processing.}
As discussed in Section~\ref{sec:implementation}, in our proof-of-concept implementation miners maintain separate queues for the traffic they process. Miners on the mainchain maintain two queues; one is for payment transactions, as they are generated as quota of the overall traffic, and are allocated a quota of each block size (in our experiment, it's 30\%). The other is for all other mainchain traffic (i.e., all other transactions, except Proofs-of-Retrievability when running with \sysname). The second queue enforces a higher priority for sync-transactions, and all the other ones are FIFO. Miners on the sidechain committee have another queue for sidechain traffic (i.e., Proofs-of-Retrievability transactions). In each round, the round leader forms a block of the transactions in the queue corresponding to the respective mainchain or sidechain. Each block is filled with transactions up to its capacity, or until queues of a specific chain are empty.

% that's all folks
\end{document}